\documentclass[11pt]{article}
\usepackage{makeidx}
\usepackage{euscript}
\usepackage{amsmath,amssymb,amsfonts}
\usepackage{dsfont}
\usepackage{latexsym}
\usepackage{subcaption}
\usepackage{graphicx}
\usepackage{fancybox}
\usepackage{fullpage}
\usepackage[backref=page]{hyperref}
\usepackage{paralist}
\usepackage{color}
\usepackage{wrapfig}
\usepackage{tikz}
\usetikzlibrary{decorations.pathreplacing}
\usepackage{setspace}
\usepackage{algorithm}
\usepackage[noend]{algpseudocode}
\usepackage[framemethod=tikz]{mdframed}
\usepackage{xspace}
\usepackage{pgfplots}
\usepackage{framed}
\usepackage{bbm}
\usepackage{subcaption}
\usepackage{thmtools}
\usepackage{thm-restate}
\pgfplotsset{compat=1.5}
\allowdisplaybreaks

\newenvironment{proof}{\noindent{\bf Proof : \ }}{\hfill$\Box$\par\medskip}

\newtheorem{theorem}{Theorem}[section]
\newtheorem{corollary}[theorem]{Corollary}
\newtheorem{lemma}[theorem]{Lemma}

\newtheorem{definition}[theorem]{Definition}

\newtheorem{problem}[theorem]{Problem}

\newenvironment{proofof}[1]{\begin{trivlist} \item {\bf Proof
#1:~~}}
  {\qed\end{trivlist}}

\newcommand{\namedref}[2]{\hyperref[#2]{#1~\ref*{#2}}}
\newcommand{\thmlab}[1]{\label{thm:#1}}
\newcommand{\thmref}[1]{\namedref{Theorem}{thm:#1}}
\newcommand{\lemlab}[1]{\label{lem:#1}}
\newcommand{\lemref}[1]{\namedref{Lemma}{lem:#1}}

\newcommand{\corlab}[1]{\label{cor:#1}}
\newcommand{\corref}[1]{\namedref{Corollary}{cor:#1}}
\newcommand{\seclab}[1]{\label{sec:#1}}
\newcommand{\secref}[1]{\namedref{Section}{sec:#1}}

\newcommand{\figlab}[1]{\label{fig:#1}}
\newcommand{\figref}[1]{\namedref{Figure}{fig:#1}}

\newcommand{\deflab}[1]{\label{def:#1}}
\newcommand{\defref}[1]{\namedref{Definition}{def:#1}}

\newcommand{\PPr}[1]{\ensuremath{\mathbf{Pr}\left[#1\right]}}

\newcommand{\Ex}[1]{\ensuremath{\mathbb{E}\left[#1\right]}}
\newcommand{\EEx}[2]{\ensuremath{\underset{#1}{\mathbb{E}}\left[#2\right]}}
\renewcommand{\O}[1]{\ensuremath{{O}\left(#1\right)}}

\newcommand{\D}{\mathcal{D}}
\newcommand{\eps}{\epsilon}

\DeclareMathOperator*{\polylog}{polylog}
\DeclareMathOperator*{\poly}{poly}
\DeclareMathOperator*{\Bin}{Bin}

\newcommand{\ignore}[1]{}

\newcommand{\kl}[2]{D_\text{KL}\left(#1\|#2\right)}
\newcommand{\Ber}[1]{\text{Ber}\left(#1\right)}

\newif\ifnotes\notestrue 
\ifnotes
\newcommand{\samson}[1]{\textcolor{purple}{{\bf (Samson:} {#1}{\bf ) }} \marginpar{\tiny\bf
             \begin{minipage}[t]{0.5in}
               \raggedright S:
            \end{minipage}}}
\newcommand{\jakab}[1]{\textcolor{blue}{{\bf (Jakab:} {#1}{\bf ) }} \marginpar{\tiny\bf
             \begin{minipage}[t]{0.5in}
               \raggedright
            \end{minipage}}}  
\else
\newcommand{\samson}[1]{}
\newcommand{\jakab}[1]{}
\fi

\makeatletter
\renewcommand*{\@fnsymbol}[1]{\textcolor{blue}{\ensuremath{\ifcase#1\or *\or \dagger\or \ddagger\or
 \mathsection\or \mathparagraph\or \triangledown\or \|\or **\or \dagger\dagger
   \or \ddagger\ddagger \else\@ctrerr\fi}}}
\makeatother

\providecommand{\email}[1]{\href{mailto:#1}{\nolinkurl{#1}\xspace}}

\definecolor{darkcerulean}{rgb}{0.03, 0.27, 0.49}
\definecolor{mahogany}{rgb}{0.75, 0.25, 0.0}
\definecolor{darkblue}{rgb}{0.0, 0.0, 0.55}
\definecolor{darkpastelgreen}{rgb}{0.01, 0.75, 0.24}
\definecolor{darkgreen}{rgb}{0.0, 0.2, 0.13}
\definecolor{darkgoldenrod}{rgb}{0.72, 0.53, 0.04}
\definecolor{darkred}{rgb}{0.55, 0.0, 0.0}
\hypersetup{
     colorlinks   = true,
     citecolor    = darkpastelgreen,
		 linkcolor		= blue,
		 urlcolor			= blue
}

\title{Noisy Boolean Hidden Matching with Applications}
\date{}

\author{
\hspace{0.22in}
Michael Kapralov\thanks{School of Computer and Communication Sciences, EPFL. 
E-mail: \email{michael.kapralov@epfl.ch}}
\and
Amulya Musipatla\thanks{School of Computer Science, Carnegie Mellon University.
E-mail: \email{amusipat@andrew.cmu.edu}}
\and
Jakab Tardos\thanks{School of Computer and Communication Sciences, EPFL. 
E-mail: \email{jakab.tardos@epfl.ch}}
\hspace{0.22in}
\and
David P. Woodruff\thanks{School of Computer Science, Carnegie Mellon University. 
E-mail: \email{dwoodruf@cs.cmu.edu}}
\and
Samson Zhou\thanks{School of Computer Science, Carnegie Mellon University. 
E-mail: \email{samsonzhou@gmail.com}}
}

{\makeatletter
	\gdef\xxxmark{%
		\expandafter\ifx\csname @mpargs\endcsname\relax 
		\expandafter\ifx\csname @captype\endcsname\relax 
		\marginpar{xxx}
		\else
		xxx 
		\fi
		\else
		xxx 
		\fi}
	\gdef\xxx{\@ifnextchar[\xxx@lab\xxx@nolab}
	\long\gdef\xxx@lab[#1]#2{{\bf [\xxxmark #2 ---{\sc #1}]}}
	\long\gdef\xxx@nolab#1{{\bf [\xxxmark #1]}}
}

\newcommand{\YES}{{\bf YES}~}
\newcommand{\NO}{{\bf NO}~}

\begin{document}
\maketitle
\begin{abstract}
The Boolean Hidden Matching (BHM) problem, introduced in a seminal paper of Gavinsky et. al. [STOC'07],  has played an important role in the streaming lower bounds for graph problems such as triangle and subgraph counting, maximum matching, MAX-CUT, Schatten $p$-norm approximation, maximum acyclic subgraph, testing bipartiteness, $k$-connectivity, and cycle-freeness. The one-way communication complexity of the Boolean Hidden Matching problem on a universe of size $n$ is $\Theta(\sqrt{n})$, resulting in $\Omega(\sqrt{n})$ lower bounds for constant factor approximations to several of the aforementioned graph problems.  The related (and, in fact, more general) Boolean Hidden Hypermatching (BHH) problem introduced by Verbin and Yu [SODA'11] provides an approach to proving higher lower bounds of $\Omega(n^{1-1/t})$ for integer $t\geq 2$. Reductions based on Boolean Hidden Hypermatching generate distributions on graphs with connected components of diameter about $t$, and basically show that long range exploration is hard in the streaming model of computation with adversarial arrivals. 

In this paper we introduce a natural variant of the BHM problem, called noisy BHM (and its natural noisy BHH variant), that we use to obtain higher than $\Omega(\sqrt{n})$ lower bounds for approximating several of the aforementioned problems in graph streams when the input graphs consist only of components of diameter bounded by a fixed constant. 
We also use the noisy BHM problem to show that the problem of classifying whether an underlying graph is isomorphic to a complete binary tree in insertion-only streams requires $\Omega(n)$ space, which seems challenging to show using BHM or BHH alone. 
\end{abstract}
\setcounter{page}{0}
\newpage
\section{Introduction}
The \emph{streaming model} of computation has emerged as a popular model for processing large datasets. 
In \emph{insertion-only} streams, sequential updates to an underlying dataset arrive over time and are permanent, while in \emph{dynamic} or \emph{turnstile} streams, the sequential updates to the dataset may be subsequently reversed by future updates.  
As many modern large datasets are most naturally represented by graphs  (e.g., social networks, protein interaction networks, or communication graphs in network monitoring) there has been a substantial amount of recent interest in graph algorithms on data streams, both on insertion-only streams, e.g.,~\cite{FeigenbaumKMSZ05,FeigenbaumKMSZ08,VerbinY11,Kapralov13,KapralovKS14,CrouchS14,KapralovKS15,PazS17,BuryGMMSVZ19} and dynamic streams, e.g.,~\cite{AhnGM12a,AhnGM12b,KapralovW14,BhattacharyaHNT15,AssadiKLY16,ChitnisCHM15,ChitnisCEHMMV16,KapralovLMMS17,NelsonY19}. 

The Boolean Hidden Matching (BHM) problem~\cite{KerenidisR06,GavinskyKW06,Bar-YossefJK08, GavinskyKKRW08} is an important problem in communication complexity that has been a major tool for showing hardness of approximation in the streaming model for a variety of graph problems, such as triangle counting~\cite{KallaugherP17,KallaugherKP18}, maximum matching~\cite{AssadiKL17,EsfandiariHLMO18,BuryGMMSVZ19}, MAX-CUT~\cite{KapralovKS15,KoganK15,KapralovKSV17}, and maximum acyclic subgraph~\cite{GuruswamiVV17}. 
In this problem, Alice is given a binary vector $x$ of length $n$ and Bob is given a matching $M$ on $[n] = \{1, 2, \ldots, n\}$ of size $\alpha n/2$ for a small positive constant $\alpha\le 1$, as well as a binary vector $w$ of length $\alpha n/2$ of labels for the edges of $M$. 
Under the promise that either $Mx\oplus w=0^{\alpha n/2}$ or $Mx\oplus w=1^{\alpha n/2}$, the goal is for Alice to send a message of minimal length, so that Bob can determine which of the two cases holds, with probability at least $\frac{2}{3}$.  Here, and in the rest of the paper, we use $M$ to denote both the matching and its edge incidence matrix, so that $Mx$ is a vector indexed by edges $e=(u, v)\in M$ such that $(Mx)_e=x_u \oplus x_v$.

Observe that if Bob determines the parity $(Mx)_e$ of any edge $e$ in the matching, then Bob can check whether $(Mx)_e\oplus w_e=0$ or $(Mx)_e\oplus w_e=1$. 
Thus, it suffices for Alice to send the parities of enough vertices so that with probability at least $\frac{2}{3}$, the parities of both vertices of some edge are revealed to Bob. 
Through a straightforward birthday paradox argument, it follows that $\O{\sqrt{n}}$ communication suffices.   
\cite{KerenidisR06,GavinskyKW06,GavinskyKKRW08} showed that this plain text protocol is essentially tight: 
\begin{theorem}
\thmlab{thm:bhm:old}
\cite{KerenidisR06,GavinskyKW06,GavinskyKKRW08}
Any randomized one-way protocol for the Boolean Hidden Matching problem that succeeds with probability at least $\frac{2}{3}$ requires $\Omega(\sqrt{n})$ bits of communication.
\end{theorem}
One of the main weaknesses of the Boolean Hidden Matching problem is that its $\Omega(\sqrt{n})$ communication complexity is not strong enough to characterize the complexity of more difficult problems. 

To address this shortcoming, Verbin and Yu proposed the Boolean Hidden Hypermatching problem (BHH)~\cite{VerbinY11}, in which Alice is given a binary vector $x$ of length $n=kt$ and Bob is given a hypermatching $M$ on $[n]$ in which all hyperedges contain $t$ vertices, as well as a binary vector $w$ of length $k$. 
Under the promise that either $Mx\oplus w=0^k$ or $Mx\oplus w=1^k$, the goal is for Alice to send some message of minimal length, so that Bob can determine which of the two cases holds, with probability at least $\frac{2}{3}$. 
Whereas BHM has complexity $\Omega(\sqrt{n})$, \cite{VerbinY11} showed that BHH has complexity $\Omega(n^{1-1/t})$:
\begin{theorem}
\thmlab{thm:BHH}
\cite{VerbinY11}
Any randomized protocol that succeeds with probability at least $\frac{2}{3}$ for the Boolean Hidden Hypermatching problem with hyperedges that contain $t$ vertices requires Alice to send $\Omega(n^{1-1/t})$ bits of communication.
\end{theorem}

Boolean Hidden Hypermatching has been used to show stronger lower bounds for cycle counting~\cite{VerbinY11}, maximum matching~\cite{AssadiKL17,EsfandiariHLMO18,BuryGMMSVZ19}, Schatten $p$-norm approximation~\cite{LiW16}, testing biconnectivity, cycle-freeness and bipartiteness~\cite{HuangP19}, MAX-CUT~\cite{KapralovKS15,KoganK15,KapralovKSV17}. The hard distributions on  input graphs that are generated by these reductions typically produce a union of connected components of diameter $\Theta(t)$, where distinguishing between the \YES and \NO cases of the input distribution intuitively requires exploring rather long paths (of length comparable to the diameter of these components). In this work we give  higher than $\Omega(\sqrt{n})$ lower bounds  that at the same time generate graphs with connected components of bounded diameter, and therefore exploit a different source of hardness  (similarly, we prove higher than $\Omega(n^{1-1/t})$ lower bounds for several of the above problems on graphs with components of diameter bounded by $O(t)$). 

In particular, we note that a major weakness of BHH is that this problem only yields lower bounds of $\Omega(n)$ when hyperedge size satisfies $t=\Omega(\log n)$. 
Consequently in the resulting reductions, quantities such as the diameter of the graph or the size of the largest clique often also grow as $\Omega(\log n)$, which prevents the usage of BHH in showing hardness of approximation for specific classes of graphs, such as graphs with bounded diameter or bounded clique number.


\subsection{Our Contributions}
We first introduce a natural parametrization of  BHM/BHH.
\begin{definition}[Noisy Boolean Hidden Matching]
The $p$-Noisy Boolean Hidden Matching problem is a two party communication problem, with parameters $p\in[0,1]$, $\alpha\in(0,1]$, and $n$.
\begin{itemize}
\item
Alice receives a binary string $x\in\{0,1\}^{n}$ of length $n$. 
\item
Bob receives a matching $M$ of size $\alpha n/2$ on $[n]$, as well as a binary vector of labels of length $\alpha n/2$. In the \YES instance Bob's labels are the true parities of the matching edges, that is $Mx$. In the \NO instance Bob's labels are $Mx$ plus some independent random noise $(\Ber{p})$ in each coordinate.
\end{itemize}
Then the goal is for Alice to send a message of minimal length, so that Bob can distinguish between the \YES and \NO cases with probability at least $2/3$.
\end{definition}
Observe that setting $p=1$ recovers the original Boolean Hidden Matching problem. 
For significantly smaller values of $p$, the Hamming distance between \YES and \NO instance labels decreases correspondingly. 
Thus while BHM can be viewed as a gap promise problem, the Noisy Boolean Hidden Matching problem essentially parametrizes the gap size. 
Now it should be apparent that the previously discussed plain text protocol of Alice sending the parities of $\Theta(\sqrt{n})$ vertices should fail for sufficiently small $p$. 
By birthday paradox arguments, the parities for $\Theta(\sqrt{n})$ vertices correspond to the observation of parities for roughly $\Theta(1)$ edges. 
However for $p=o(1)$, Alice and Bob already know the parities of most of the edges, because the vectors generated in the $\YES$ and $\NO$ cases corresponding to the possible edge labels only differ in $o(n)$ coordinates. 
Thus any message of length $\O{\sqrt{n}}$ sent by Alice is unlikely to be helpful to Bob. 
Indeed, we show that the communication complexity of the $p$-Noisy Boolean Hidden Matching problem is generally $\Omega\left(\sqrt{\frac{n}{p}}\right)$. 
%
%
More generally, we define the $p$-Noisy Boolean Hidden Hypermatching problem as a means to find tradeoffs between the noise $p$, the complexity of the problem, and the size of hyperedges $t$ (see \secref{sec:noisy:bhh}). 

\begin{theorem}
For $p=\Omega\left(\frac{1}{n}\right)$, any randomized one-way protocol that succeeds with probability at least $\frac{2}{3}$ for the $p$-Noisy Boolean Hidden Hypermatching problem on hyperedges with $t$ vertices requires $\Omega\left(n^{1-1/t}p^{-1/t}\right)$ bits of communication.
\end{theorem}

Through the flexibility of the $p$-Noisy Boolean Hidden Hypermatching problem, we show hardness of approximation for graph problems on a parametrized family of inputs. 
We first use $p$-Noisy Boolean Hidden Hypermatching to show hardness of approximation for MAX-CUT in the streaming model on graphs whose connected components have bounded size, which is a significant obstacle for reductions to BHH. 
Unlike reductions from BHM~\cite{KapralovKS15}, our methods can show nearly linear lower bounds for approximation close to $1$ even in this setting:
\begin{restatable}{theorem}{thmmaxcutbhhparam}
\thmlab{thm:maxcut:bhh:param}
Let $2\le t\le n/10$ be an integer. 
For $p\in\left[\frac{128t}{n},\frac{1}{2}\right]$, any one-pass streaming algorithm that outputs a $\left(1+\frac{p}{14t}\right)$-approximation to the MAX-CUT with probability at least $\frac{2}{3}$ requires $\Omega\left(n^{1-1/t}p^{-1/t}\right)$ space, even for graphs with components of size bounded by $4t$. 
\end{restatable}
Similarly, we show hardness of approximation for maximum matching in the streaming model better than $\Omega(\sqrt{n})$ on graphs whose connected components have bounded size, which is again challenging for reductions from either BHH or BHM. 
\begin{restatable}{theorem}{thmmatchbhhparam}
\thmlab{thm:match:bhh:param}
Let $2\le t\le n/10$ be some integer. 
For $p\in\left[\frac{128t}{n},\frac{1}{2}\right]$, any one-pass streaming algorithm that outputs a $\left(1+\frac{p}{6t}\right)$-approximation to the maximum matching with probability at least $\frac{2}{3}$ requires $\Omega\left(n^{1-1/t}p^{-1/t}\right)$ space, even for graphs with connected components of size bounded by $\O{t}$.
\end{restatable}

By comparison for graphs with connected components with constant bounded sizes, existing lower bounds for both MAX-CUT and maximum matching only show that $\Omega(n^C)$ space is required, for some constant $C\in(0,1)$ bounded away from $1$.

Our third graph streaming lower bound proves hardness of approximation for maximum acyclic subgraph in the streaming model. 
\cite{GuruswamiVV17} showed that an $\frac{8}{7}$-approximation requires $\Omega(\sqrt{n})$ space through a reduction from BHM, but it was not evident how their reduction could be generalized to BHH, due to its hyperedge structure. 
Instead, we use our $p$-Noisy Boolean Hidden Matching communication problem to show a fine-grained lower bound for the maximum acyclic subgraph problem with tradeoffs between approximation guarantee and space. 
Independently, \cite{assadi2020multipass} showed a lower bound that $(1-\eps)$-approximation requires $\Omega(n^{1-O(\eps^c)})$ space for a fixed constant $c>0$ through a reduction from their one-or-many cycles communication problem. 

\begin{restatable}{theorem}{thmmas}
\thmlab{thm:mas}
Let $2\le t\le n/10$ be some integer. For $p\in\left[\frac{128}{n},\frac{1}{2}\right]$, any one-pass streaming algorithm that outputs a $\left(1+\frac{p}{22}\right)$-approximation to the maximum acyclic subgraph problem with probability at least $\frac{2}{3}$ requires $\Omega\left(\sqrt{\frac{n}{p}}\right)$ space. 
\end{restatable}

Finally, we introduce and study the {\em graph classification problem} in data streams. 
Here the goal is for the streaming algorithm to output \YES if the input graph belongs to a specified isomorphism class, and output \NO otherwise.  
Graph isomorphism is one of the most fundamental problems in computer science; it asks whether there exists an isomorphism between two given graphs, or more specifically, whether there exists a bijection between the vertices of the two graphs that preserves edges, i.e., the image of adjacent vertices remain adjacent. A special important case of this problem is graph isomorphism on tree graphs. 

We show hardness of graph classification for complete binary trees on insertion-only streams. This is a class of tree graphs for which we do not know how to prove hardness using any other technique, showing our communication problem may be useful for ultimately resolving the general graph classification problem on data streams. 
In this setting, a stream of insertions and deletions of edges in an underlying graph with $n$ vertices arrives sequentially and the task is to classify whether the resulting graph is isomorphic to a complete binary tree on $n$ vertices. 
Although it is possible to produce a lower bound of $\Omega(\sqrt{n})$ space from BHM, it is not evident that reductions from BHH can produce stronger lower bounds. 
Instead, we use our $p$-Noisy Boolean Hidden Matching communication problem to show a lower bound of $\Omega(n)$ space for graph classification of complete binary trees even on insertion-only streams. 

\begin{restatable}{theorem}{thmtreenoisybhm}
\thmlab{thm:tree:noisy:bhm}
Any randomized algorithm on insertion-only streams that correctly classifies with probability at least $\frac{3}{4}$ whether an underlying graph is a complete binary tree uses $\Omega(n)$ space.
\end{restatable}

In fact, we show more general parametrized space lower bounds for testing on streams whether an underlying graph is a complete binary tree or $\eps$-far from being a complete binary tree, where $\eps$-far is defined as follows:

\begin{definition}
We say that a graph $G=(V,E_1)$ is $\epsilon$ far from another graph $H=(V,E_2)$ on the same vertex set $V$ if at least $\epsilon\cdot|V|$ edge insertions or deletions are required to get from $G$ to $H$, i.e., $|(E_1\setminus E_2)\cup(E_2\setminus E_1)|\ge\epsilon\cdot|V|$. 
\end{definition}

\begin{restatable}{theorem}{thmtreenoisybhmpt}
\thmlab{thm:tree:noisy:bhm:pt}
For $\epsilon\in[\tfrac{512}n,\tfrac12]$, any randomized algorithm on insertion-only streams that classifies, with probability at least $\frac{3}{4}$, whether an underlying graph is a complete binary tree or $\epsilon/16$-far from a complete binary tree uses $\Omega\left(\sqrt{n/\epsilon}\right)$ space.
\end{restatable}

\paragraph{Related work.} 
Several communication problems inspired by the Boolean Hidden (Hyper)Matching problem have recently been used in the literature to prove tight lower bounds for the single pass or sketching complexity of several graph problems (e.g., ~\cite{KapralovKSV17,KapralovK19} for the MAX-CUT problem,~\cite{KallaugherKP18} for subgraph counting, in~\cite{GuruswamiVV17,GuruswamiT19,CGV20,ChouGSV21} for general CSPs). 
The recent work of~\cite{assadi2020multipass} gives multipass streaming lower bounds for the space complexity of the aforementioned one-or-many cycles communication problem, which is tightly connected to BHH, extending many of the abovementioned single pass lower bounds to the multipass setting. 
Although (to the best of our knowledge) property testing for graph isomorphism on streams have not been previously studied, there is an active line of work, e.g.~\cite{ FeigenbaumKMSZ05, FeigenbaumKMSZ08, SunW15, GuhaMT15, MonemizadehMPS17, HuangP19, CzumajF0S20, AssadiN21} studying property testing on graphs implicitly defined through various streaming models.

\subsection{Overview}
We now outline the analysis of the Boolean Hidden Matching problem~\cite{KerenidisR06,GavinskyKW06,GavinskyKKRW08}, and describe the key differences in our approach.

\paragraph{A natural extension of BHM analysis and why it fails.} Recall that in BHM, Alice receives a binary vector $x\in \{0,1\}^n$, and sends Bob a message $m$ of $c$ bits.
Letting $A\subseteq \{0, 1\}^n$ denote the indicator of a `typical' message, one shows that for a `typical' matching $M$ of size $\alpha n$, $\alpha\in (0, 1/2)$,  Bob's posterior distribution $q$ on $Mx$ conditioned on the message received from Alice is close to uniform, which in turn implies that Bob cannot distinguish between $w=Mx$ and $w=Mx\oplus 1^n$. The approach of~\cite{GavinskyKW06} upper bounds the total variation distance from $q$ to the uniform distribution via the $\ell_2$ distance: this lets one upper bound the $\ell_2$ distance to uniform in Fourier domain, and then use Cauchy-Schwarz to obtain the required bound of $\Omega(\sqrt{n})$ on the size $c$ of Alice's message.

Since we are trying to obtain an $\Omega(\sqrt{n/p})$ lower bound for the $p$-noisy Boolean Hidden Matching problem, it becomes clear that one can no longer compare Bob's posterior to the uniform distribution (the bound of $\Omega(\sqrt{n})$ is tight here).
Instead, one would like to compare the distribution of Bob's labels in the \YES case to the same distribution in the \NO case.  
A natural approach here is to compare Bob's posterior distribution $q$ to the noisy version of the posterior, and relating these two distributions to the Fourier transform of Alice's message using the noise operator $T_\rho$ for an appropriate choice of the parameter $\rho$.  Interestingly, however, this approach fails: one can verify that the expected $\ell_2^2$ (over the randomness of the matching $M$) distance is too large\footnote{If $f:\{0, 1\}^n \to \{0, 1\}$ is the characteristic function of a typical message from Alice, the $\ell_2^2$ distance between the two distributions is, up to appropriate scaling factors, equal to the sum of squares of Fourier coefficients $\widehat{f}(s)$, scaled by a $(1-(1-2p)^{|s|})\approx p\cdot |s|$ factor. While this factor gives us exactly the required $p$ contribution for Fourier coefficients $s$ of  Hamming weight $1$, the contribution of higher weight terms is much larger, precluding the analysis.}.

\paragraph{Our approach.}  We provide an information theoretic proof for the complexity of $p$-Noisy Boolean Hidden Hypermatching\footnote{In an earlier version of the paper we claimed that the difficulties with the standard $\ell_2$-based Fourier analytic approach can be circumvented using a related method that relies on KL divergence instead of the $\ell_2$ norm. We discovered an error in that proof and it is therefore omitted from the present version.}.
Informally, we create a product distribution over $\frac{1}{p}$ instances of BHH, where the $i$-th instance has size roughly distributed as $\Bin\left(n,p\right)$ and is a \YES instance with probability $\frac{1}{2}$, and a \NO instance with probability $\frac{1}{2}$. 
The distribution of the sizes for each instance allows us to match the distribution of ``flipped'' edges in the $p$-Noisy BHH distribution. 
We would like to use the product distribution to claim that any protocol that solves $p$-Noisy BHH on a graph with $n$ vertices must also solve $\frac{1}{p}$ instances of BHH on graphs with $np$ vertices; however, communication complexity is not additive so we use (conditional) information complexity. 

We bound the $1$-way conditional information cost of any protocol, correct on our distribution of interest, by first using a message compression result of \cite{JainPY12} to relate the $1$-way conditional information cost of a single instance to the $1$-way distributional communication complexity of the instance. Our conditional information cost is conditioned on Bob's inputs. We note that recent work \cite{AssadiKL17} only bounds external information cost and it does not seem immediate how to derive the same lower bound for conditional information cost from their work \cite{assadi}. 
We then observe that any protocol which solves the $p$-Noisy BHH must also solve a constant fraction of the $\frac{1}{p}$ instances of BHH with size $\Omega(np)$, by distributional correctness (note it need not solve all $\frac{1}{p}$ instances).  
Using the conditional information cost in a direct sum argument, the communication complexity of the protocol must be at least $\frac{1}{p}\cdot\Omega\left((np)^{1-1/t}\right)=\Omega\left(n^{1-1/t}p^{-1/t}\right)$, which lower bounds the communication.  

\paragraph{Lower bound applications.}
The remaining lower bounds are quite simple; they result from natural generalizations of existing BHM or BHH reductions to the $p$-Noisy BHM or BHH. 
For the MAX-CUT problem, \cite{KapralovKS15} give a reduction from BHM that creates a connected component with eight edges for each edge $m_i$ in the matching $M$. 
In the case $(Mx)_i\oplus w_i=0$, the resulting graph is bipartite, so that the max cut induced by the connected component is $8$. 
In the case $(Mx)_i\oplus w_i=1$, the resulting graph is not bipartite, so that the max cut induced by the connected component is at most $7$ (see \figref{fig:maxcut:reduction}). 
Thus the max cut for $Mx\oplus w=0^n$ has size $4n$ and the max cut for $Mx\oplus w=1^n$ has size at most $7n/2$, so any sufficiently small constant factor approximation algorithm to the max cut requires $\Omega(\sqrt{n})$ space. 
Observe that the same reduction from $p$-Noisy BHM also works, although since we set $\alpha=1/2<1$, we also have to consider components corresponding to unmatched vertices. Due to the fact that only a $p$ fraction of the components change their contribution from $8$ to $7$ in the \NO case, our reduction works for $\left(1+\Theta(p)\right)$-approximation.
We give a similar argument for $p$-Noisy BHH, which allows parametrization of the connected component size. 

To show hardness of approximation for the maximum matching problem, we use a reduction similar to that of~\cite{BuryGMMSVZ19}. However, we reduce from $p$-noisy BHH. We represent each coordinate of Alice's input with a single edge, and represent each hyperedge of Bob with two cdisjoint cliques. Supports of the cliques are defined in such a way that the resulting connected components have even size exactly if $(Mx)_i\oplus w_i=0$, in which case they are perfectly matchable (see \figref{fig:maxmatching:reduction}). In the noisy ({\bf NO}) case, however, a $p$ fraction of these components of size $O(t)$ will have odd size, which leads to an overall $(1+\Theta(p/t))$ factor loss in the size of the maximum matching.

To show hardness of approximation for maximum acyclic subgraph, we use a reduction by~\cite{GuruswamiVV17}.  
For each $i\in[n]$, the case $(Mx)_i\oplus w_i=1$ corresponds to an isolated subgraph with eight edges that contains no cycle, so that its maximum acyclic subgraph has size eight. 
However, the case $(Mx)_i\oplus w_i=0$ creates an isolated subgraph with eight edges that contains a cycle, so that its maximum acyclic subgraph has size seven (see \figref{fig:mas:reduction}). 
Thus if $Mx\oplus w=0^n$ (the \YES case), then all subgraphs corresponding to matching edges contribute only $7$ to the maximum acyclic subgraph. However, in the \NO case some of these contribute $8$, increasing the total size of the maximum acyclic subgraph by a factor $(1+\Theta(p))$.

To show hardness of classifying whether an underlying graph is a complete binary tree, we use a gadget by~\cite{EsfandiariHLMO18} that embeds BHM into the bottom layer of a binary tree. 
For each $i\in[n]$ where $n$ is assumed to be a power of two, the case $(Mx)_i\oplus w_i=0$ creates two paths of length two, which can be used to extend the binary tree to an additional layer at two different nodes. 
However, the case $(Mx)_i\oplus w_i=1$ creates a path of length one and a path of length three, which results in a non-root node having degree two (see \figref{fig:bhh:paths}). 
Thus the resulting graph is a complete binary tree if and only if $Mx\oplus w=0^n$. 
While BHM requires $\Omega(\sqrt{n})$ space to distinguish whether $Mx\oplus w=0^n$ or $Mx\oplus w=1^n$, the $p$-Noisy version of BHM demands $\Omega\left(\sqrt{\frac{n}{p}}\right)$ space to distinguish between the \YES and \NO cases.
Now for $p=\Theta\left(\frac{1}{n}\right)$, we have $Mx\oplus w\neq 0^n$ with high probability in the \NO case.
Hence, the graph classification problem also solves $p$-Noisy BHM in this regime of $p$ and requires $\Omega(n)$ space.


\section{Preliminaries}
We use the notation $[n]$ to denote the set $\{1,2,\ldots,n\}$. 
We use $\poly(n)$ to denote a fixed constant degree polynomial in $n$ and $\frac{1}{\poly(n)}$ to denote some arbitrary degree polynomial in $n$ corresponding to the choice of constants in the algorithms. 
We use $\polylog(n)$ to denote polylogarithmic factors of $n$. 
We say an event occurs with high probability if it occurs with probability at least $1-\frac{1}{\poly(n)}$. 
For $x,y\in\{0,1\}$, we use $x\oplus y$ to denote the sum of $x$ and $y$ modulo $2$.

We define $\alpha$-approximation for maximization problems for one-sided error (as opposed to two-sided errors):
\begin{definition}[$\alpha$-Approximation for Maximization Problems]
For a parameter $\alpha\ge 1$, we say that an algorithm $\mathcal{A}$ is an $\alpha$-approximation algorithm for a maximization problem with optimal value $\mathsf{OPT}$ if $\mathcal{A}$ outputs some value $X$ with $X\le\mathsf{OPT}\le\alpha X$.  
\end{definition}

	\section{Noisy Boolean Hidden Hypermatching}
	\seclab{sec:noisy:bhh}
	We start with the definition of the $p$-noisy Boolean Hidden Hypermatching problem ($p$-noisy BHH), then given a simple protocol for it and show that this protocol is asympototically tight.

	\begin{definition}
	A $t$-hypermatching on $[n]$ is a collection of disjoint subsets of $[n]$, each of size $t$, which we call hyperedges.
	\end{definition}
The noisy Boolean Hidden Hypermatching problem is:
	\begin{definition}
	\deflab{def:bhh}
		For $p\in [0, 1]$ and integer $t\geq 2$ the $p$-Noisy Boolean Hidden $t$-Hypermatching ($p$-noisy BHH) is a one way two party communication problem:
		\begin{itemize}
			\item Alice gets $x\in\{0,1\}^n$ uniformly at random.
			\item Bob gets $M$, a $t$-hypermatching of $[n]$ with $\alpha n/t$ hyperedges, for a constant $\alpha\in (0, 1]$. $M$ is considered to be the $\alpha n/t\times n$ incidence matrix of the matching hyperedges. That is, the $i^\text{th}$ row of $M$ has ones corresponding to the vertices of the $i^\text{th}$ matching hyperedge, and zeros elsewhere.
			  
			 \item Bob also receives labels for each hyperedge $m_i$. In the \YES case  Bob receives edge labels $w=Mx$ (the true parities of each the matching hyperedges with respect to $x$). In the \NO case (the noisy case) Bob receives labels $w\oplus z\in \{0, 1\}^M$, where each $z_i$ is an independent $\Ber{p}$ variable.
			 
			 
			\item Output: Bob must determine whether the communication game is in the \YES or \NO case.
		\end{itemize}
	\end{definition}
	
	
		\paragraph{A protocol for $p$-noisy BHH.} We begin by presenting a simple protocol for solving Boolean hidden hypermatching, which turns out to be nearly asymptotically optimal: for some $c\ge t$ Alice sends a set $S$ of $c$ random bits of $x$ to Bob. Bob then takes each hyperedge that is fully supported on $S$, and verifies that his label is the true parity. If all such labels reflect the true parity, Bob guesses the \YES case, while if there is a discrepancy Bob guesses the \NO case.
		
		It is clear that the only way the protocol can fail is if in the \NO case no hyperedge is simultaneously supported on $S$ and mislabeled. Let us call the event that the $i^\text{th}$ hyperedge, $m_i$ is supported on $S$ and mislabeled $\mathcal E_i$. For each hyperedge of $M$, the probability of being supported on $S$ is $\binom ct/\binom nt \ge c^t/e^tn^t$, while the probability of being mislabeled is independently $p$. Therefore,
		$$\mathbb P(\mathcal E_i)\ge p\cdot\left(\frac{c}{en}\right)^t.$$
		Since the events $\mathcal E_i$ are negatively correlated we have that
		\begin{align*}
		    \mathbb P\left(\bigcup_{i=1}^{\alpha n/t}\mathcal E_i\right)&\ge1-\prod_{i=1}^{\alpha n/t}\left(1-\mathbb P(\mathcal E_i)\right)\\
		    &\ge1-\left(1-p\cdot\left(\frac{c}{en}\right)^t\right)^{\alpha n/t}\\
		    &\ge1-\exp\left(-\frac{p\alpha n}{t}\cdot\left(\frac{c}{en}\right)^t\right)
		\end{align*}
	Therefore, if $t\le n/10$ and $c=\Omega(n^{1-1/t}(p\alpha)^{-1/t})$, this probability is an arbitrarily large constant, and Bob can distinguish between the \YES and \NO cases with high constant probability.

\subsection{Communication Complexity of Noisy BHH}

In this section, we establish our lower bound on the communication complexity of Noisy BHH. 
The core of the proof is a reduction to Boolean Hidden Hypermatching, whose complexity we defined in \thmref{thm:BHH}. The advantage of this approach is that we are able to extend the result of the previous section to the case of $\alpha > 1/2$, and in particular to $\alpha = 1$ -- where Bob receives a perfect hypermatching. However, in this setting we need to be careful about the parity of the noise applied to the labels in the \NO case.

\begin{definition}\label{def:ZPN}
We define the distribution $\mathcal Z_p^n$, as identical to the independent vector of $\Ber{p}$ variables of length $n$, but with the condition that the number of ones is even:
$$V\sim\Ber{p}^n\ \Longrightarrow\ V\Big|2|w_H(V)\sim\mathcal Z_p^n$$
where $w_H$ denotes the Hamming weight. More formally, if $Z\sim\mathcal Z_p^n$ we have
$$\mathbb P(Z=z)=\frac{\mathbbm1(2|w_H(z))p^{w_H(z)}(1-p)^{n-w_H(z)}}{\sum_{z'\in\{0,1\}^n}\mathbbm1(2|w_H(z'))p^{w_H(z')}(1-p)^{n-w_H(z')}}.$$

Furthermore, let $|\mathcal Z_p^n|$ denote the distribution of the Hamming weight of a variable from $\mathcal Z_p^n$.
\end{definition}

We are now able to define the distributional communication problem of $p$-Noisy Perfect Boolean Hidden Hypermatching:

\begin{definition}\label{def:PBHH}
For $p\in[0,1]$ and integer $t\ge 2$ the $p$-Noisy Perfect Boolean Hidden $t$-Hypermatching ($p$-Noisy PBHH) is a one-way two-party communication problem:
\begin{itemize}
    \item Alice gets $x\in\{0,1\}^n$.
    \item Bob gets $M$, a perfect $t$-hypermatching of $[n]$ (that is $n/t$ hyperedges).
    \item Bob also receives edge labels $w\in\{0,1\}^{n/t}$. With probability $1/2$, we are in the \YES case, and the edge labels satisfy $Mx= w$. With probability $1/2$, we are in the \NO case and the edge labels satisfy $Mx\oplus z = w$, where $Z\sim\mathcal Z_p^{n/t}$.
    \item Output: Bob must determine whether the communication game is in the \YES or \NO case.
\end{itemize}

\end{definition}

We prove a lower bound on the communication and information complexity of $p$-Noisy BHH by considering the uniform input distribution: We are in the \YES and \NO cases with probability $1/2$ each; $x$ is sampled uniformly at random from all $n$-length bit strings and $M$ is sampled uniformly at random from all $t$-hypermatchings of size $\alpha n/t$ on $[n]$. We replace the variables $x$, $w$, and $z$ with the {\emph random} variables $X$, $W$ and $Z$ respectively.
%
%
We first define quantities from information theory in order to show a direct sum theorem for internal information. 
\begin{definition}[Entropy and conditional entropy]
The \emph{entropy} of a random variable $X$ is defined as 
\[H(X):=\sum_x p(x)\log\frac{1}{p(x)},\]
where $p(x)=\PPr{X=x}$. 
The \emph{conditional entropy} of $X$ with respect to a random variable $Y$ is defined as 
\[H(X|Y)=\EEx{y}{H(X|Y=y)}.\]
\end{definition}

\begin{definition}[Mutual information and conditional mutual information]
The \emph{mutual information} between random variables $A$ and $B$ is defined as 
\[I(A;B)=H(A)-H(A|B)=H(B)-H(B|A).\]
The \emph{conditional mutual information} between $A$ and $B$ conditioned on a random variable $C$ is defined as 
\[I(A;B|C)=H(A|C)-H(A|B,C).\]
\end{definition}

\begin{definition}[Communication Complexity]
Given a distributional one-way communication problem on inputs in $\mathcal X\times\mathcal Y$ distributed according to $\mathcal D$, and a one-way protocol $\Pi$, let $\Pi(X)$ denote the message of Alice under input $X$. Then the \emph{communication complexity} of $\Pi$ is the maximum length of $\Pi(X)$, over all inputs $X$ and private randomness $r$:
\[\mathsf{CC}(\Pi):=\max_{X\in\text{supp}(\mathcal{D}), \textrm{ private randomness } r}|\Pi_r(X)|,\]
where $\Pi_r(X)$ is the length of Alice's message on input $X$ with private randomness $r$. 
\end{definition}

\begin{definition}[Internal Information Cost]
Given a distributional one-way communication problem on inputs in $\mathcal X\times\mathcal Y$ distributed according to $\mathcal D$, and a one-way protocol $\Pi$, let $\Pi(X)$ denote the message of Alice under input $X$.
Then we define the \emph{internal information cost} of $\Pi$ with respect to some other distribution $\mathcal D'$ to be 
\[\mathsf{IC}_{\mathcal D'}(\Pi):=I_{\mathcal{D'}}(\Pi(X);X|Y, R),\]
where $\Pi(X)$ denotes the message sent from Alice to Bob in the protocol $\Pi$ on input $X$. Here $R$ is the public randomness. 
\end{definition}
Here $\mathcal{D}$ is the correctness distribution for the protocol $\Pi$ and $\mathcal{D'}$ is a distribution for measuring information.




The following well-known theorem (e.g., Lemma 3.14 in \cite{BravermanR14}) shows that the communication complexity of any protocol is at least the (internal) information cost of the protocol. 
\begin{theorem}
\label{thm:info:cc}
Given a distributional one-way communication problem on inputs in $\mathcal X\times\mathcal Y$ distributed according to $\mathcal D$, and a one-way protocol $\Pi$, let $\Pi(X)$ denote the message of Alice under input $X$.
Then, if $\text{supp}(\mathcal D')\subseteq\text{supp}(\mathcal D)$,
\[\mathsf{CC}(\Pi)\ge\mathsf{IC}_{\mathcal D'}(\Pi).\]
\end{theorem}
Note that we include the condition $\text{supp}(\mathcal D')\subseteq\text{supp}(\mathcal D)$ since our distributional communication problem is defined as the maximum message length over inputs in the support of $\mathcal D$, and thus $I_{\mathcal{D'}}(\Pi(X);X|Y) \leq H(\Pi(X) | Y) \leq H(\Pi(X)) \leq {\bf E}_{X \sim \mathcal D'}|\Pi(X)|$, where $|\Pi(X)|$ is the length of Alice's message on random input $X$, which in turn is at most $\max_{x \in \text{supp}(\mathcal D'), \textrm{ private randomness } r} |\Pi_r(x)| \leq \max_{x \in \text{supp}(\mathcal D), \textrm{ private randomness }r} |\Pi_r(x)| = \mathsf{CC}(\Pi)$. Here $|\Pi_r(x)|$ denotes the length of Alice's message on input $x$ with private randomness $r$.  

The core of our proof is to show that any protocol that succeeds in solving the $p$-Noisy PBHH with high constant probability has high internal information cost. This statement is then reduced to a statement about the information cost of a distributional version of the standard BHH problem, through a method similar to that of \cite{Bar-YossefJKS04}.

For completeness we state this distributional version of BHH:

\begin{definition}\label{def:BHH}
For an {\emph{even}} integer $n$, and integer $t\ge2$, Boolean Hidden $t$-Hypermatching (BHH) is a one-way two-party communication problem:
\begin{itemize}
    \item Alice gets $X\in\{0,1\}^n$ uniformly at random.
    \item Bob gets $M$, a perfect $t$-hypermatching of $[n]$ (that is, $n/t$ hyperedges) uniformly at random.
    \item Bob also receives edge labels $W\in\{0,1\}^{n/t}$. With probability $1/2$, we are in the \YES case, and Bob's edge labels satisfy $MX=W$. With probability $1/2$, we are in the \NO case and Bob's edge labels satisfy $MX\oplus W=1^{n/t}$.
    \item Output: Bob must determine whether the communication game is in the \YES or \NO case.
\end{itemize}
 We call this input distribution of BHH $\mathcal D$, and call the distributions when conditioning on the \YES and \NO cases $\mathcal D_{\bf YES}$ and $\mathcal D_{\bf NO}$, respectively.
\end{definition}

\begin{lemma}[Lemma 3.4 in~\cite{JainPY12} with $t=1$]
\label{lem:info:comm}
Given a distributional one-way communication problem on inputs in $\mathcal X\times\mathcal Y$ distributed according to $\mathcal D'$. Suppose there exists a one-way communication protocol $\Pi'$ succeeding with probability $1-\delta$ with
$$IC_{\mathcal D'}(\Pi')\le c.$$
Then for any $\epsilon>0$ there exists some other one-way communication protocol $\Pi$ with
$$CC(\Pi)\le\frac{c+5}{\epsilon}+O\left(\log\frac1\epsilon\right),$$
and succeeding with probability $1-\delta-6\epsilon$.
\end{lemma}

\begin{corollary}\label{cor:BHH-D-IC-bound}
Any randomized protocol $\Pi$ that succeeds with probability at least $\frac{2}{3}$ over the distribution $\mathcal{D}$ for the Boolean Hidden Hypermatching problem with hyperedges that contain $t$ vertices requires internal information cost $\Omega(n^{1-1/t})$ when measured on the $D$ distribution.
\end{corollary}

\begin{proof}
It is known, e.g.~\cite{AssadiKL17}, that the communication complexity of this distributional version of BHH is $\Omega(n^{1-1/t})$ for any protocol that succeeds with probability bounded away from $1/2$. However, any protocol with internal information cost $o(n^{1-1/t})$, when combined with Lemma~\ref{lem:info:comm} for a small constant $\epsilon$, would result in a better protocol for BHH; a contradiction.
\end{proof}


\begin{corollary}[Conditional information lower bounds for BHH]
\corlab{cor:info:bhh}
Any randomized protocol $\Pi$ that succeeds with probability at least $\frac{2}{3}$ over the distribution $\mathcal{D}$ for the Boolean Hidden Hypermatching problem with hyperedges that contain $t$ vertices requires internal information cost $\Omega(n^{1-1/t})$ when measured on the $\mathcal D_{\bf YES}$ distribution. That is:
$$I_{\mathcal D_{\bf YES}}(X;\Pi(X)|M,W)=\Omega(n^{1-1/t}).$$
\end{corollary}
\begin{proof}
By Corollary~\ref{cor:BHH-D-IC-bound} we know that
\begin{align*}
    I_{\mathcal D}(X;\Pi(X)|M,W)=\Omega(n^{1-1/t})
\end{align*}
We decompose this into information complexity with respect to the $\mathcal D_{\bf YES}$ and $\mathcal D_{\bf NO}$ distributions. Let $b$ be the single bit denoting whether the input is in the \YES or \NO case. Since $b$ has entropy $1$, adding or removing the conditioning on $b$ amounts to at most a change of $1$ in the mutual information.
\begin{align*}
    I_{\mathcal D}(X;\Pi(X)|M,W)&\le 1+I_{\mathcal D}(X;\Pi(X)|M,W,b)\\
    &=1+\tfrac12I_{\mathcal{D_{\bf YES}}}(X;\Pi(X)|M,W,b)+\tfrac12I_{\mathcal{D_\NO}}(X;\Pi(X)|M,W,b)
\end{align*}
Note that $X$ (and consequently $\Pi(X)$) follow the same distribution in both the \YES and \NO cases. Furthermore, note that BHH flips each bit of the edge labels deterministically in the \NO case. Therefore, in both the \YES and \NO cases, revealing $b$ and $W$ amounts to revealing the true edge parities. Thus, two terms ($I_{\mathcal D_{\bf YES}}$ and $I_{\mathcal{D}_{\textbf{NO}}}$) are equal:
\begin{align*}
    I_{\mathcal D}(X;\Pi(X)|M,W)&\le1+I_{\mathcal D_{\bf YES}}(X;\Pi(X)|M,W,b)\\
    &\le2+I_{\mathcal D_{\bf YES}}(X;\Pi(X)|M,W).
\end{align*}
The statement of the corollary then follows.
\end{proof}

We define one more distributional one-way communication problem to bridge the gap between BHH and $p$-Noisy PBHH. It consists of $1/(2p)$ instances of BHH, of which exactly one or none are in the \NO case. Furthermore all the instances of BHH are of variable size:

\begin{definition}\label{def:VBHH}
Let $p\in\left(\frac{t}{2n},\frac{1}{2}\right]$, where we assume for simplicity that $Q:=1/2p$ is an integer. For $t\ge 2$ integer, the Variable Size, $Q$-Copy Boolean Hidden $t$-Hypermatching ($\overline{\text{VBHH}}$) is a one-way two-party communication problem:
\begin{itemize}
    \item $S_1,\ldots,S_Q$ are drawn independently from $|\mathcal Z_p^{n/t}|$ (recall Definition~\ref{def:ZPN}), and $r$ is drawn independently, uniformly from $[Q]$.
    \item In the \YES case, Alice and Bob get $Q$ independent copies of BHH distributed according to $\mathcal D_{\bf YES}$, where the $i^{\text{th}}$ copy is on $t\cdot S_i$ vertices.
    \item In the \NO case, Alice and Bob again get $Q$ independent copies of BHH, where the $i^\text{th}$ copy has size $t\cdot S_i$. All copies are distributed according to $\mathcal D_{\bf YES}$, except the $r^\text{th}$ copy, which is distributed according to $\mathcal D_{\bf NO}$.
    \item Output: Bob must determine whether the communication game is in the \YES or \NO case.
\end{itemize}
We call this input distribution of $\overline{\text{VBHH}}$ $\overline{\mathcal D}$, and call the distributions when conditioning on the \YES and \NO cases $\overline{\mathcal D}_{\bf YES}$ and $\overline{\mathcal D}_{\bf NO}$ respectively.
\end{definition}

\begin{theorem}\label{thm:BHH-IC-lower-bound}
Any randomized protocol $\Pi$ that succeeds with probability at least $\frac{4}{5}$ over the distribution $\overline{\mathcal{D}}$ for the $\overline{\text{VBHH}}$ parameters $n$, $2\le t\le n/100$ and $p\in\left(\frac{t}{2n},\frac{1}{2}\right]$, has internal information cost $\Omega(n^{1-1/t}p^{-1/t})$ when measured on the $\overline{\mathcal D}_{\bf YES}$ distribution. That is:
$$I_{\overline{\mathcal D}_{\bf YES}}(X;\Pi(X)|M,W)=\Omega(n^{1-1/t}p^{-1/t}).$$
\end{theorem}

\begin{proof}
Note that in the context of $\overline{\text{VBHH}}$ $X$, $M$, and $W$ are vectors of length $Q$, where ($X_i$, $M_i$, $W_i$) are individual instances of inputs of BHH. Let $\Pi$ be a protocol that succeeds with probability $4/5$ on $\overline{\mathcal D}$. Then
\begin{align}
I_{\overline{\mathcal D}_{\bf YES}}(X;\Pi(X)|M, W)&=\sum_{i=1}^QI_{\overline{\mathcal D}_{\bf YES}}(X_i;\Pi(X)|M,W,X_{<i})\nonumber\\
&\ge\sum_{i=1}^QI_{\overline{\mathcal D}_{\bf YES}}(X_i;\Pi(X)|M,W),\label{eq:mutual-info-sum}
\end{align}
where the mutual information only decreases by removing the conditioning on $X_{<i}$ because $X_i$ is independent of $X_{<i}$.

Now, we know that $\Pi$ succeeds with probability at least $4/5$. Since $r$ is distributed uniformly in $[Q]$, on at least a constant fraction of the possible values of $r$, $\Pi$ must succeed with probability at least $3/4$ for that specific value of $r$. Let $R\subseteq[Q]$ be the set of such values of $r$, so that $|R|=\Omega(Q)$.

Consider a term corresponding to $i\in R$:
\begin{align}
    I_{\overline{\mathcal D}_{\bf YES}}(X_i;\Pi(X)|M,W).\label{eq:summand}
\end{align}
We will demonstrate that there exists a protocol for solving BHH as defined in Definition~\ref{def:BHH} with probability at least $3/4$, with internal information cost equal to the quantity in Equation~\ref{eq:summand}. By \corref{cor:info:bhh} this will allow us to lower bound Equation~\ref{eq:summand}.

The protocol is as follows: Alice and Bob are given inputs of BHH distributed according to $\mathcal D$, which we will call $X_i$, $M_i$, $W_i$. Alice and Bob will then generate $Q-1$ more instances of BHH from the distribution $\mathcal D_{\bf YES}$, called $(X_j,M_j,W_j)$, for $j=1,\ldots,i-1,i+1,\ldots,Q$. $M_j$ and $W_j$ are always generated using public randomness, while $X_j$ is generated using the private randomness of Alice. This allows Alice to have access to $M_j$ and $W_j$, which in turn allows $X_j$ to be generated in such a way that $(X_j,M_j,W_j)$ is indeed distributed according to $\mathcal D_{\bf YES}$. Alice then simply sends the message $\Pi(X_1,\ldots,X_Q)=\Pi(X)$. Since all instances of BHH other than the $i^\text{th}$ one are in the \YES case, the distribution $(X,M,W)$ is exactly $\overline{\mathcal D}$ conditioned on $r=i$; furthermore, it is in the \YES case exactly if $(X_i, M_i, W_i)$ is in the \YES case. Since $i\in R$, $\Pi$ succeeds with probability at least $3/4$, allowing Alice and Bob to succeed in solving BHH with the same probability.

Finally, we must address the fact that the above protocol succeeds with probability $3/4$ on an instance of BHH with variable size ($t\cdot S_i\sim t\cdot|\mathcal Z_p^{n/t}|$). One can easily verify that for large enough $n$ $\mathbb P(S_i<pn/(2t))$ is at most $1/100$. Therefore, the protocol must succeed with probability at least $3/4-1/100\ge2/3$ conditioned on $S_i\ge pn/(2t)$. Finally, this means that the protocol must succeed with probability at least $2/3$ for at least one specific such setting of $S_i$. Therefore, by \corref{cor:info:bhh},
$$I_{\overline{\mathcal D}_{\bf YES}}(X_i;\Pi(X)|M,W)\ge\Omega\left(\left(t\cdot\frac{pn}{2t}\right)^{1-1/t}\right)=\Omega\left((pn)^{1-1/t}\right).$$

\noindent
Therefore, returning to Equation~\ref{eq:mutual-info-sum}, we have that
\begin{align*}
I_{\overline{\mathcal D}_{\bf YES}}(X;\Pi(X)|M, W)&\ge\sum_{i=1}^QI_{\overline{\mathcal D}_{\bf YES}}(X_i;\Pi(X)|M,W)\ge\Omega\left(Q(pn)^{1-1/t}\right)=\Omega\left(n^{1-1/t}p^{-1/t}\right).
\end{align*}
\end{proof}

Finally, we lower bound the communication complexity of $p$-Noisy PBHH, by a reduction from $\overline{\text{VBHH}}$.

\begin{theorem}
\thmlab{thm:pbhh}
For $t\le n/10$, and $p\in\left(\frac{t}{2n},\frac{1}{2}\right]$, the communication complexity of $p$-Noisy Perfect Boolean Hidden Hypermatching with success probability $5/6$ is $\Omega(n^{1-1/t} p^{-1/t})$ bits.
\end{theorem}
\begin{proof}
Suppose there exists a one-way communication protocol solving $p$-Noisy PBHH with probability at least $5/6$. We will construct a communication protocol for $\overline{\text{VBHH}}$ with the same parameters $n$, $t$ and $p$.

Let Alice and Bob receive $X, S$ and $M, W, S$ respectively, distributed according to $\overline{\mathcal D}$ as defined in Definition~\ref{def:VBHH}. One can easily verify that the total number of vertices in all subinstances of BHH, that is $\sum t\cdot S_i$, is at most $n$ with probability at least $99/100$ for large enough $n$. We call this event $\mathcal E$; if $\mathcal E$ is not true, Alice and Bob fail the protocol. From now on we condition on $\mathcal E$.

Let us label the vertices from all subinstances of BHH by
$$V=\{(i,j)|i\in[Q],\ j\in[t\cdot S_i]\}.$$
Let $\phi$ be an embedding of $V$ into $[n]$ chosen uniformly at random using public randomness. Since we assumed the event $\mathcal E$, such an embedding exists, and since we are using public randomness to generate it, both Alice and Bob have access to $\phi$.

Alice and Bob generate their inputs for $p$-Noisy PBHH, which they will call $(X^*, M^*, W^*)$. We call the elements of $[n]$ covered by $\phi$ (that is $\phi(V)$), $A^+$; we call the remaining elements, not covered by $\phi$, $A^-$. We define the inputs of Alice and Bob separately on the two parts of $[n]$.

Bob generates his input ($M^*$ and $W^*$) on $A^+$ as follows: For all $i\in[Q]$, and all hyperedges $\{j_1,j_2,\ldots,j_t\}\in M_i$, Bob generates a hyperedge $\{\phi(i,j_1),\ldots,\phi(i,j_t)\}\subseteq[n]$, which inherits the label of $\{j_1,\ldots,j_t\}\in M_i$. Alice generates her input as follows: For each $(i,j)\in V$, Alice sets $X^*({\phi(i,j)})$ to $X_i(j)$. 

On $A^-$, Alice and Bob use public randomness to generate both of their inputs, allowing them to be dependent. Alice sets $X^*$ uniformly at random on $A^-$, while Bob adds a uniformly random perfect hypermatching of $A^-$ to $M^*$. Finally, the labels of this perfect matching are set to reflect the true parities of $X^*$, that is $M^*_iX^*\oplus W^*_i=0$ for any $M^*_i$ matching hyperedge supported on $A^-$.

Since $\phi$ was chosen uniformly randomly, the distribution of $(X^*, M^*, W^*)$ conforms to what is prescribed in Definition~\ref{def:PBHH}. Indeed, $X^*$ and $M^*$ are uniformly random. Furthermore, if the original instance of $\overline{\text{VBHH}}$ is in the \YES case, the edge labels always reflect the true parity of $X^*$ on the hyperdege (and thus the constructed instance of $p$-Noisy PBHH is also in the \YES case) due to the parity assumptions on the edge labels $W$. When the original instance of $\overline{\text{VBHH}}$ is in the \NO case, a random subset of the edge labels are flipped. This subset has size $S_r\sim|\mathcal Z_p^{n/t}|$, exactly as in the \NO case of $p$-Noisy PBHH.

In conclusion, if Alice and Bob can solve their newly constructed instance of $p$-Noisy PBHH with probability $5/6$, they can solve the original instance of $\overline{\text{VBHH}}$ with probability at least
$$5/6\cdot\mathcal P(\mathcal E)\ge4/5.$$
Therefore, by Theorem~\ref{thm:BHH-IC-lower-bound}, and Theorem~\ref{thm:info:cc} the protocol requires space $\Omega\left(n^{1-1/t}p^{-1/t}\right)$.
\end{proof}

\section{Applications}
In this section, we give a number of applications for the $p$-Noisy BHM and $p$-Noisy BHH problems. 
Through a reduction from $p$-Noisy BHH, we first show hardness of approximation for max cut in the streaming model on graphs whose connected components have bounded size, which is a significant obstacle for reductions to BHM. 
Unlike reductions from BHM, our reduction from $p$-Noisy BHM can still show nearly linear lower bounds in this setting. 

We then show hardness of approximation for maximum matching in the streaming model better than $\Omega(\sqrt{n})$ on graphs whose connected components have bounded size. 
Again our reduction displays flexibility beyond what is offered by either BHM or BHH for both parametrization of component size and tradeoffs between approximation guarantee and space complexity. 

Finally, we show hardness of approximation for maximum acyclic subgraph. 
\cite{GuruswamiVV17} show a $\Omega(\sqrt{n})$ lower bound for $\frac{8}{7}$-approximation, but the reduction does not readily translate into a stronger lower bound from BHH, due to the hyperedge structure of BHH. 
We use our $p$-Noisy BHM problem to give a fine-grained lower bound that provides tradeoffs between the approximation guarantee and the required space.

\subsection{MAX-CUT}
Recall the following definition of the MAX-CUT problem. 
\begin{problem}[MAX-CUT]
Given an unweighted graph $G=(V,E)$, the goal is to output the maximum of the number of edges of $G$ that cross a bipartition, over all bipartitions of $V$, i.e., $\max_{P\cup Q=V,P\cap Q=\emptyset}|E\cap(P\times Q)|$. 
\end{problem}
To show hardness of approximation of MAX-CUT in the streaming model, where the edges of the underlying graph $G$ arrive sequentially, we use a reduction similar to~\cite{KapralovKS15}, who gave a reduction from BHM that creates a connected component with $8$ edges for each edge $m_i$ in the input matching $M$ from BHM. 
The key property of the reduction of \cite{KapralovKS15} is that for $(Mx)_i\oplus w_i=0$, the connected component corresponding to $m_i$ is bipartite and its induced max cut has size $8$, but for $(Mx)_i\oplus w_i=1$, the connected component is not bipartite and its induced max cut has size at most $7$ (see \figref{fig:maxcut:reduction}). 
Therefore, the max cut for $Mx\oplus w=0^n$ has size $8n$ and the max cut for $Mx\oplus w=1^n$ has size at most $7n$, which gives the desired separation for a constant factor approximation algorithm. 

We instead reduce from Noisy Boolean Hidden Hypermatching. 
Suppose Alice is given a binary vector $x$ of length $n=2kt$ and Bob is a given a hypermatching $M$ of size $n/2t=k$ on $n$ vertices, where each edge contains $t$ vertices, so that $\alpha=\frac{1}{2}$. 
Bob also receives a vector $w$ of length $k$, generated according to the \YES or \NO case of the $p$-noisy BHH (see~\defref{def:bhh}). 
To distinguish between the two cases:
\begin{itemize}
\item 
Alice creates the four vertices $a_i$, $b_i$, $c_i$, and $d_i$ for each $i\in[n]$ corresponding to a coordinate of $x$. 
\item
If $x_i=0$, then Alice adds the edges $(a_i,b_i)$, $(c_i,d_i)$, and $(a_i,d_i)$, but if $x_i=1$, then Alice instead adds the edges $(a_i,b_i)$, $(c_i,d_i)$, and $(a_i,c_i)$. 
\item
For each hyperedge $m_i=(j_{i,1},\ldots,j_{i,t})$ of $M$, with $j_{i,s}\le j_{i,s+1}$, if $w_i=0$ Bob adds the edges $(d_{j_{i,s}},a_{j_{i,s+1}})$ for $s\in[t-1]$, and the edge $(d_{j_{i,t}},a_{j_{i,1}})$. 
Otherwise if $w_i=1$, then Bob instead adds the edges $(d_{j_{i,s}},a_{j_{i,s+1}})$ for $s\in[t-1]$ and the edge $(d_{j_{i,t}},b_{j_{i,1}})$. 
\end{itemize}
By design, the connected component of the graph corresponding to $m_i$ is bipartite if and only if $(Mx)_i\oplus w_i=0$. 
Hence, the max cut has size $4t$ if $(Mx)_i\oplus w_i=0$ and size at most $4t-1$ otherwise, i.e., when $(Mx)_i\oplus w_i=1$.

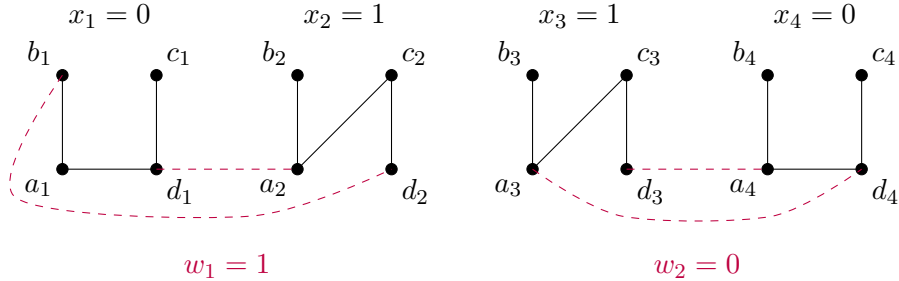
\begin{figure}[!htb]
\centering
\begin{tikzpicture}[scale=2.5]
\draw[fill] (0,0) circle (0.03);
\draw[fill] (0.5,0) circle (0.03);
\draw[fill] (0.5,0.5) circle (0.03);
\draw[fill] (0,0.5) circle (0.03);
\draw (0,0.5)--(0,0);
\draw (0,0)--(0.5,0);
\draw (0.5,0)--(0.5,0.5);
\node[below left] at (0,0){$a_1$};
\node[above left] at (0,0.5){$b_1$};
\node[below right] at (0.5,0){$d_1$};
\node[above right] at (0.5,0.5){$c_1$};
\node[above] at (0.25,0.7){$x_1=0$};

\draw[fill] (1.25+0,0) circle (0.03);
\draw[fill] (1.25+0.5,0) circle (0.03);
\draw[fill] (1.25+0.5,0.5) circle (0.03);
\draw[fill] (1.25+0,0.5) circle (0.03);
\draw (1.25+0,0.5)--(1.25+0,0);
\draw (1.25+0,0)--(1.25+0.5,0.5);
\draw (1.25+0.5,0)--(1.25+0.5,0.5);
\node[below left] at (1.25+0,0){$a_2$};
\node[above left] at (1.25+0,0.5){$b_2$};
\node[below right] at (1.25+0.5,0){$d_2$};
\node[above right] at (1.25+0.5,0.5){$c_2$};
\node[above] at (1.25+0.25,0.7){$x_2=1$};

\draw[fill] (2.5+0,0) circle (0.03);
\draw[fill] (2.5+0.5,0) circle (0.03);
\draw[fill] (2.5+0.5,0.5) circle (0.03);
\draw[fill] (2.5+0,0.5) circle (0.03);
\draw (2.5+0,0.5)--(2.5+0,0);
\draw (2.5+0,0)--(2.5+0.5,0.5);
\draw (2.5+0.5,0)--(2.5+0.5,0.5);
\node[below left] at (2.5+0,0){$a_3$};
\node[above left] at (2.5+0,0.5){$b_3$};
\node[below right] at (2.5+0.5,0){$d_3$};
\node[above right] at (2.5+0.5,0.5){$c_3$};
\node[above] at (2.5+0.25,0.7){$x_3=1$};

\draw[fill] (3.75+0,0) circle (0.03);
\draw[fill] (3.75+0.5,0) circle (0.03);
\draw[fill] (3.75+0.5,0.5) circle (0.03);
\draw[fill] (3.75+0,0.5) circle (0.03);
\draw (3.75+0,0.5)--(3.75+0,0);
\draw (3.75+0,0)--(3.75+0.5,0);
\draw (3.75+0.5,0)--(3.75+0.5,0.5);
\node[below left] at (3.75+0,0){$a_4$};
\node[above left] at (3.75+0,0.5){$b_4$};
\node[below right] at (3.75+0.5,0){$d_4$};
\node[above right] at (3.75+0.5,0.5){$c_4$};
\node[above] at (3.75+0.25,0.7){$x_4=0$};

\draw [purple, dashed] plot [smooth] coordinates {(0,0.5) (-0.25,-0.15) (1,-0.25) (1.25+0.5,0)};
\draw [purple, dashed] (0.5,0)--(1.25+0,0);
\node [below] at (0.875,-0.4){\color{purple}{$w_1=1$}};

\draw [purple, dashed] plot [smooth] coordinates {(2.5+0,0) (2.5+0.475,-0.25) (2.5+1.275,-0.25) (3.75+0.5,0)};
\draw [purple, dashed] (2.5+0.5,0)--(2.5+1.25+0,0);
\node [below] at (2.5+0.875,-0.4){\color{purple}{$w_2=0$}};
\end{tikzpicture}
\caption{Example of reduction from Noisy Boolean Hidden Matching (Noisy BHH with $t=2$) to MAX-CUT. 
Solid lines added by Alice, dashed purple lines added by Bob. 
Here we have $m_1=(1,2)$ and $m_2=(3,4)$. 
Note that $x_1\oplus x_2\oplus w_1=0$ and has a cycle of even length while $x_3\oplus x_4\oplus w_2=1$ has a cycle of odd length.}
\figlab{fig:maxcut:reduction}
\end{figure}

\thmmaxcutbhhparam*
\begin{proof}
Recall that by design, the connected component of the graph corresponding to $m_i$ is bipartite if and only if $(Mx)_i\oplus w_i=0$. 
Hence, the max cut has size $4t$ if and only if $(Mx)_i\oplus w_i=0$. 
Otherwise, the connected component has an odd cycle, so that the max cut has size at most $4t-1$. 
Connected components corresponding to unmatched elements of $[n]$ have max cut size $3$.


Since $p\ge\frac{128t}{n}$, the probability that at least $\frac{np}{4t}$ of the $\frac{n}{2t}$ edge labels are flipped is at least $\frac{99}{100}$. 
Thus, in the \NO case, the max cut has size at most 
\[(4t)\cdot\left(\frac{n}{2t}-\frac{np}{4t}\right)+(4t-1)\cdot\frac{np}{4t}+\frac{n}{2}\cdot3=\frac{7n}{2}-\frac{np}{4t},\]
with probability at least $\frac{99}{100}$. 
However, in the \YES case the max cut has size $\frac{n}{2t}(4t)+\frac{n}{2}\cdot3=\frac{7n}{2}$. 
Thus, since $p\le 1$, 
\[\frac{7n/2}{7n/2-\frac{np}{4t}}\ge 1+\frac{p}{14t}.\]
Therefore, any $\left(1+\frac{p}{14t}\right)$-approximation to the MAX-CUT requires $\Omega\left(n^{1-1/t}p^{-1/t}\right)$ space. 
Finally, recall that each of the $\O{\frac{n}{t}}$ connected components has at most $4t$ edges.
\end{proof}

\subsection{MAX-MATCHING}
Recall the following definition of the MAX-MATCHING problem. 
\begin{problem}[Maximum Matching]
Given an unweighted graph $G=(V,E)$, the goal is to output the maximum size of a matching, that is a set of disjoint edges.
\end{problem}
To show hardness of approximation of MAX-MATCHING in the streaming model, we use a reduction similar to~~\cite{BuryGMMSVZ19}. We use a reduction from noisy BHH. In the reduction, each matching hyperedge $m_i$ corresponds to a gadget consisting of two connected components of size $O(t)$. The key observation is that if $m_ix\oplus w_i=0$, then the two components are of even size, and can be perfectly matched, but if $m_ix\oplus w_i=1$, then the two components are of odd size, and cannot be perfectly matched.

We reduce from Noisy Boolean Hidden Hypermatching with $\alpha = 1/2$.
Suppose Alice is given a binary vector $x$ of length $n$ and Bob is a given a hypermatching $M$ of size $n/2t$ on the vertex-set $[n]$, where each edge contains $t$ vertices.
Bob also receives a vector $w$ of length $n/2t$, generated according to the \YES or \NO case of the $p$-noisy BHH (see~\defref{def:bhh}). 
The protocol to distinguish between the two cases slightly differs based on the parity of $t$, but the core idea is the same. We begin by considering odd $t$.
\begin{itemize}
\item 
Alice creates the four vertices $a_i$, $b_i$, $c_i$, and $d_i$ for each $i\in[n]$ corresponding to a coordinate of $x$. 
\item
If $x_i=0$, then Alice adds the single edges $(a_i,b_i)$, but if $x_i=1$, then Alice instead adds the edge $(c_i,d_i)$.
\item
For each hyperedge $m_i=(j_{i,1},\ldots,j_{i,t})$ of $M$, Bob adds a single vertex $e_i$. Then, if $w_i=0$ Bob adds a clique between $(b_{j_i,1}, b_{j_{i,2}}, \ldots, b_{j_{i,t}})$ and another clique between $(d_{j_i,1}, d_{j_{i,2}}, \ldots, d_{j_{i,t}}, e_i)$.
If $w_i=1$, Bob adds the same two cliques, but moving $e_i$ from the second clique to the first. Formally Bob adds the cliques $(b_{j_i,1}, b_{j_{i,2}}, \ldots, b_{j_{i,t}}, e_i)$ and $(d_{j_i,1}, d_{j_{i,2}}, \ldots, d_{j_{i,t}})$.
\end{itemize}

If $t$ is even, the protocol is very similar, but we state it here for completeness:

\begin{itemize}
\item 
Alice does the same thing as in the previous case: she creates the four vertices $a_i$, $b_i$, $c_i$, and $d_i$ for each $i\in[n]$ corresponding to a coordinate of $x$. 
\item
If $x_i=0$, then Alice adds the single edges $(a_i,b_i)$, but if $x_i=1$, then Alice instead adds the edge $(c_i,d_i)$.
\item
For each hyperedge $m_i=(j_{i,1},\ldots,j_{i,t})$ of $M$, Bob adds two new vertices $e_i$ and $f_i$. Then, if $w_i=0$ Bob adds a clique between $(b_{j_i,1}, b_{j_{i,2}}, \ldots, b_{j_{i,t}})$ and Bob adds another clique between $(d_{j_i,1}, d_{j_{i,2}}, \ldots, d_{j_{i,t}}, e_i, f_i)$.
If $w_i=1$, Bob adds the same two cliques, but moving $e_i$ from the second clique to the first. Formally Bob adds the cliques $(b_{j_i,1}, b_{j_{i,2}}, \ldots, b_{j_{i,t}}, e_i)$ and $(d_{j_i,1}, d_{j_{i,2}}, \ldots, d_{j_{i,t}}, f_i)$.
\end{itemize}

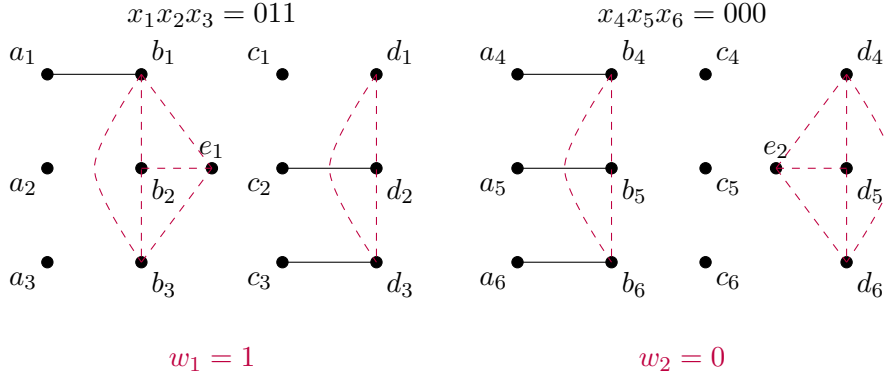
\begin{figure}[!htb]
\centering
\begin{tikzpicture}[scale=2.5]
\draw[fill] (0,0) circle (0.03);
\draw[fill] (0.5,0) circle (0.03);
\draw[fill] (0.5,0.5) circle (0.03);
\draw[fill] (0,0.5) circle (0.03);
\draw[fill] (0.5,-0.5) circle (0.03);
\draw[fill] (0,-0.5) circle (0.03);
\draw (0.5,0.5)--(0,0.5);
\node[below left] at (0,0){$a_2$};
\node[above left] at (0,0.5){$a_1$};
\node[below right] at (0.5,0){$b_2$};
\node[above right] at (0.5,0.5){$b_1$};
\node[below left] at (0,-0.5){$a_3$};
\node[below right] at (0.5,-0.5){$b_3$};

\node[above] at (0.875,0.7){$x_1x_2x_3=011$};

\draw[fill] (0.875, 0) circle(0.03);
\node[above] at (0.875, 0) {$e_1$};

\draw[fill] (1.25+0,0) circle (0.03);
\draw[fill] (1.25+0.5,0) circle (0.03);
\draw[fill] (1.25+0.5,0.5) circle (0.03);
\draw[fill] (1.25+0,0.5) circle (0.03);
\draw[fill] (1.25+0.5,-0.5) circle (0.03);
\draw[fill] (1.25+0,-0.5) circle (0.03);
\draw (1.25+0,0)--(1.25+0.5,0);
\draw (1.25+0,-0.5)--(1.25+0.5,-0.5);
\node[below left] at (1.25+0,0){$c_2$};
\node[above left] at (1.25+0,0.5){$c_1$};
\node[below right] at (1.25+0.5,0){$d_2$};
\node[above right] at (1.25+0.5,0.5){$d_1$};
\node[below left] at (1.25+0,-0.5){$c_3$};
\node[below right] at (1.25+0.5,-0.5){$d_3$};

\draw[fill] (2.5+0,0) circle (0.03);
\draw[fill] (2.5+0.5,0) circle (0.03);
\draw[fill] (2.5+0.5,0.5) circle (0.03);
\draw[fill] (2.5+0,0.5) circle (0.03);
\draw[fill] (2.5+0.5,-0.5) circle (0.03);
\draw[fill] (2.5+0,-0.5) circle (0.03);
\draw (2.5+0,0.5)--(2.5+0.5,0.5);
\draw (2.5+0,0)--(2.5+0.5,0);
\draw (2.5+0,-0.5)--(2.5+0.5,-0.5);
\node[below left] at (2.5+0,0){$a_5$};
\node[above left] at (2.5+0,0.5){$a_4$};
\node[below right] at (2.5+0.5,0){$b_5$};
\node[above right] at (2.5+0.5,0.5){$b_4$};
\node[below left] at (2.5+0,-0.5){$a_6$};
\node[below right] at (2.5+0.5,-0.5){$b_6$};

\draw[fill] (2.5+0.875+0.5,0) circle (0.03);
\node[above] at (2.5+0.875+0.5,0) {$e_2$};

\node[above] at (2.5+0.875,0.7){$x_4x_5x_6=000$};

\draw[fill] (3.5+0,0) circle (0.03);
\draw[fill] (3.75+0.5,0) circle (0.03);
\draw[fill] (3.75+0.5,0.5) circle (0.03);
\draw[fill] (3.5+0,0.5) circle (0.03);
\draw[fill] (3.75+0.5,-0.5) circle (0.03);
\draw[fill] (3.5+0,-0.5) circle (0.03);
\node[below right] at (3.5+0,0){$c_5$};
\node[above right] at (3.5+0,0.5){$c_4$};
\node[below right] at (3.75+0.5,0){$d_5$};
\node[above right] at (3.75+0.5,0.5){$d_4$};
\node[below right] at (3.5+0,-0.5){$c_6$};
\node[below right] at (3.75+0.5,-0.5){$d_6$};

\draw [purple, dashed] plot [smooth] coordinates {(0.5,0.5) (0.25,0) (0.5, -0.5)};
\draw [purple, dashed] (0.5,0.5)--(0.5,-0.5);
\draw [purple, dashed] (0.5,0.5)--(0.875, 0);
\draw [purple, dashed] (0.5,0)--(0.875, 0);
\draw [purple, dashed] (0.5,-0.5)--(0.875, 0);

\draw [purple, dashed] plot [smooth] coordinates {(1.75,0.5) (1.25 + 0.25,0) (1.75, -0.5)};
\draw [purple, dashed] (1.75,0.5)--(1.75,-0.5);

\node [below] at (0.875,-0.9){\color{purple}{$w_1=1$}};

\draw [purple, dashed] plot [smooth] coordinates {(2.5+0.5,0.5) (2.5 + 0.25, 0) (2.5+0.5,-0.5)};
\draw [purple, dashed] (2.5+0.5,0.5)--(2.5+0.5,-0.5);
\draw [purple, dashed] plot [smooth] coordinates {(3.75+0.5,0.5) (3.75+0.5 + 0.25, 0) (3.75+0.5,-0.5)};
\draw [purple, dashed] (3.75+0.5,0.5)--(3.75+0.5,-0.5);
\draw [purple, dashed] (2.5+0.875+0.5,0)--(3.75+0.5,0.5);
\draw [purple, dashed] (2.5+0.875+0.5,0)--(3.75+0.5,0);
\draw [purple, dashed] (2.5+0.875+0.5,0)--(3.75+0.5,-0.5);

\node [below] at (2.5+0.875,-0.9){\color{purple}{$w_2=0$}};
\end{tikzpicture}
\caption{An illustration of the graph construction for the maximum matching reduction for $t=3$. Solid lines added by Alice, dashed lines added by Bob. On the left side, we see the components corresponding to $m_1=\{1,2,3\}$. Since $m_1x\oplus w_1=1$ this subgraph contributes $4$ to the maximum matching. On the right side, we see the components corresponding to $m_2=\{4,5,6\}$. Since $m_2x\oplus w_2=0$ this subgraph contributes $5$ to the maximum matching.}
\figlab{fig:maxmatching:reduction}
\end{figure}

\thmmatchbhhparam*
\begin{proof}
In this proof we focus on the case where $t$ is odd. The case of even $t$ is similar.

Consider an edge of the hypermatching $M$: $m_i=(j_{i,1},\ldots,j_{i,t})$. Consider further all the vertices related to this hyperedge, specifically $a_j$, $b_j$, $c_j$, and $d_j$ for all $j\in\{j_{i,1},\ldots,j_{i,t}\}$ as well as $e_i$. Suppose among the values $x_{i,1},\ldots,x_{i,t}$ there are exactly $s$ ones (and correspondingly exactly $t-s$ zeros). Let vertices of the form $a_{j_{i,k}}$ corresponding $x_{j_{i,k}}=0$ make up $A_0^i$, and ones corresponding to $x_{j_{i,k}}=1$ make up $A_1^i$. Define $B_0^i$ and $B_1^i$, $C_0^i$ and $C_1^i$, as well as $D_0^i$ and $D_1^i$ analogously. Note that $|A_0^i|=|B_0^i|=|C_0^i|=|D_0^i|=t-s$ and $|A_1^i|=|B_1^i|=|C_1^i|=|D_1^i|=s$.

Consider the case when $w_i=0$. Then, by the above construction, Bob puts cliques between $B_0^\cup B_1^i$ and $D_0^i\cup D_1^i\cup\{e_i\}$. Furthermore, recall that Alice puts an edge between $a_j$ and $b_j$ when $x_j=0$ and an edge between $c_j$ and $d_j$ when $x_1$. Therefore, the connected components of this subgraph are $A_0^i\cup B_0^i\cup B_1^i$ and $C_1^i\cup D_0^i\cup D_1^i\cup\{e_i\}$, while vertices of $A_1^\cup C_0^i$ remain isolated.

The size of the first component, $A_0^i\cup B_0^i\cup B_1^i$, is $2t-s$. In the case when $s$ is even this has even size, and it is easy to verify that it can be perfectly matched, so it contributes $t-s/2$ to the total maximum matching size. On the other hand, if $s$ is odd, it contributes at most $t-(s+1)/2$. Similarly, the second component, $C_1^i\cup D_0^i\cup D_1^i\cup\{e_i\}$, has size $t+s+1$. When $s$ is even it can be perfectly matched, contributing $(t+s+1)/2$ to the maximum matching size. However, when $s$ is odd, it contributes only $(t+s)/2$. Overall, when $s$ is even (corresponding to $m_ix\oplus w_i = 0$) the maximum matching size of the subgraph corresponding to $m_i$ is $(3t+1)/2$; however, when $s$ is odd (corresponding to $m_ix\oplus w_i = 1$) it is only at most $(3t-1)/2$.

A nearly identical analysis shows that when $w_i=1$, the maximum matching size of the subgraph corresponding to $m_ix\oplus w_i = 0$ is $(3t+1)/2$ and the maximum matching size of the subgraph corresponding to $m_ix\oplus w_i = 1$ is $(3t-1)/2$. 
Finally, since $\alpha = 1/2 < 1$ there will also be vertices of $[n]$ that are not matched by any hyperedge in $M$. If some $j\in[n]$ does not participate in any matching hyperedge, the corresponding four vertices, $a_j$, $b_j$, $c_j$, and $d_j$ will have exactly one edge between them regardless of the value of $x_j$.

Let us now consider the overall maximum matching size of the graph built by Alice and Bob. In the \YES case, it is always true that $m_ix\oplus w_i=0$. Therefore, for $t$ odd, all the vertices corresponding to any hyperedge in $M$ will contribute exactly $(3t+1)/2$, while all remaining vertices will contribute $1$, for a total of
$$\frac{n}{2t}\cdot\frac{3t+1}{2}+\frac{n}{2}\cdot1=\frac{5n}{4}+\frac{n}{4t}.$$

In the \NO case, since $p\ge\frac{128t}{n}$, the probability that at least $\frac{np}{4t}$ of the $\frac{n}{2t}$ hyperedge labels are flipped is at least $\frac{99}{100}$. Each flipped label decreases the contribution of the corresponding component by $1$; thus, the maximum matching has size at most 
$$\frac{5n}{4}+\frac{n}{4t}-\frac{np}{4t},$$
with probability at least $\frac{99}{100}$. 

So the ratio between the two cases is at least
$$\frac{5n/4+n/(4t)}{5n/4+n/(4t)-np/(4t)}\ge\frac{3n/2}{3n/2-np/(4t)}\ge1+\frac{np}{6t}.$$

The analysis for even $t$ is very similar. However, since we add the extra vertex $f_i$ , all the vertices corresponding to any hyperedge $m_i$ satisfying $m_ix\oplus w_i=0$ will contribute exactly $(3t+2)/2$, while hyperedges satisfying $m_ix\oplus w_i=1$ will contribute $3t/2$, and all remaining vertices will contribute $1$. In the \YES case this is a total of
$$\frac{n}{2t}\cdot\frac{3t+2}{2}+\frac{n}{2}\cdot1=\frac{5n}{4}+\frac{n}{2t}.$$
In the \NO case, each flipped label again decreases the contribution of the corresponding component by $1$ so that the maximum matching has size at most 
$$\frac{5n}{4}+\frac{n}{2t}-\frac{np}{4t},$$
with probability at least $\frac{99}{100}$. 
Thus, the ratio between the two cases is at least
$$\frac{5n/4+n/(2t)}{5n/4+n/(2t)-np/(4t)}\ge\frac{3n/2}{3n/2-np/(4t)}\ge1+\frac{np}{6t}.$$

Therefore, any $\left(1+\frac{p}{6t}\right)$-approximation to MAX-MATCHING requires $\Omega\left(n^{1-1/t}p^{-1/t}\right)$ space. 
Finally, recall that each of the connected components has at most $2t+1$ vertices.
\end{proof}


\subsection{Maximum Acyclic Subgraph}
In this section, we study the hardness of approximation for the maximum acyclic subgraph on insertion-only streams. 
\begin{problem}[Maximum acyclic subgraph]
Given a directed graph $G=(V,E)$, the goal is to output the size of the largest acyclic subgraph of $G$, where the size of a graph is defined to be the number of edges in it. 
\end{problem}
To show hardness of approximation of maximum acyclic subgraph in the streaming model, where the directed edges of $G$ arrive sequentially, we consider the reduction of~\cite{GuruswamiVV17}, who created a subgraph with $8$ edges for each edge $m_i$ in the input matching $M$ from BHM. 
The key property of the reduction is that $(Mx)_i\oplus w_i=1$ corresponds to a subgraph with no cycles but $(Mx)_i\oplus w_i=0$ corresponds to a subgraph with a cycle (see \figref{fig:mas:reduction}). 
Hence, the former case allows $8$ edges in its maximum acyclic subgraph while the latter case only allows for $7$ edges. 
We use the same reduction from $p$-Noisy BHM with $\alpha=1/2$, so that there are $n/4$ matching edges. 

In particular, suppose Alice is given a string $x\in\{0,1\}^{2n}$ and Bob is given a matching of $[2n]$ of size $\alpha n$ with $\alpha=\frac{1}{2}$. 
To distinguish between the two cases:
\begin{itemize}
\item
Alice and Bob consider a graph on $4(2n)=8n$ vertices and associate four vertices $a_i$, $b_i$, $c_i$, and $d_i$ to each $i\in[2n]$, i.e., through a predetermined ordering of the vertices. 
\item
If $x_i=0$, then Alice creates the directed edges $(a_i,b_i)$ and $(d_i,c_i)$, where we use the convention that $(v_1,v_2)$ represents the directed edge $v_1\rightarrow v_2$. 
Otherwise if $x_i=1$, then Alice adds the directed edges $(b_i,a_i)$ and $(c_i,d_i)$.
\item
If $w_i=0$, Bob adds the four directed edges $(b_{y_i},a_{z_i})$, $(b_{z_i},a_{y_i})$, $(d_{y_i},c_{z_i})$, and $(d_{z_i},c_{y_i})$ for each matching edge $m_i=(y_i,z_i)$ of $M$ from the $p$-Noisy Boolean Hidden Matching problem. 
Otherwise if $w_i=1$, then Bob adds the four directed edges $(b_{y_i},b_{z_i})$, $(a_{z_i},a_{y_i})$, $(d_{y_i},d_{z_i})$, and $(c_{z_i},c_{y_i})$. 
\end{itemize}
By design, the subgraph corresponding to $m_i$ has eight edges, but consists of exactly one cycle when $x_{y_i}\oplus x_{z_i}\oplus w_i=0$ and zero cycles when $x_{y_i}\oplus x_{z_i}\oplus w_i=1$. 
Finally, any unmmatched subgraph contributes four edges due to Alice's construction. 

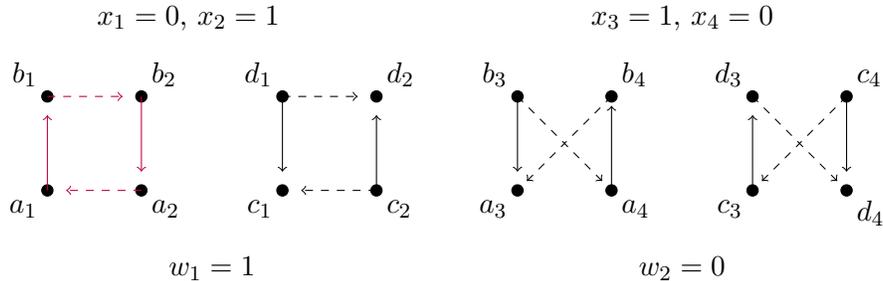
\begin{figure}[!htb]
\centering
\begin{tikzpicture}[scale=2.5]
\draw[fill] (0,0) circle (0.03);
\draw[fill] (0.5,0) circle (0.03);
\draw[fill] (0.5,0.5) circle (0.03);
\draw[fill] (0,0.5) circle (0.03);
\draw[purple, ->] (0,0)--(0,0.4);
\draw[dashed, purple, ->] (0,0.5)--(0.4,0.5);
\draw[purple, ->] (0.5,0.5)--(0.5,0.1);
\draw[dashed, purple, ->] (0.5,0)--(0.1,0);
\node[below left] at (0,0){$a_1$};
\node[above left] at (0,0.5){$b_1$};
\node[below right] at (0.5,0){$a_2$};
\node[above right] at (0.5,0.5){$b_2$};

\draw[fill] (1.25+0,0) circle (0.03);
\draw[fill] (1.25+0.5,0) circle (0.03);
\draw[fill] (1.25+0.5,0.5) circle (0.03);
\draw[fill] (1.25+0,0.5) circle (0.03);
\draw[<-] (1.25+0,0.1)--(1.25+0,0.5);
\draw[dashed,->] (1.25+0,0.5)--(1.25+0.4,0.5);
\draw[<-] (1.25+0.5,0.4)--(1.25+0.5,0);
\draw[dashed,->] (1.25+0.5,0)--(1.25+0.1,0);
\node[below left] at (1.25+0,0){$c_1$};
\node[above left] at (1.25+0,0.5){$d_1$};
\node[below right] at (1.25+0.5,0){$c_2$};
\node[above right] at (1.25+0.5,0.5){$d_2$};
\node[above] at (0.75,0.8){$x_1=0$, $x_2=1$};

\draw[fill] (2.5+0,0) circle (0.03);
\draw[fill] (2.5+0.5,0) circle (0.03);
\draw[fill] (2.5+0.5,0.5) circle (0.03);
\draw[fill] (2.5+0,0.5) circle (0.03);
\draw[->] (2.5+0,0.5)--(2.5+0,0.1);
\draw[dashed,->] (2.5+0,0.5)--(2.5+0.45,0.05);
\draw[->] (2.5+0.5,0)--(2.5+0.5,0.45);
\draw[dashed,->] (2.5+0.5,0.5)--(2.5+0.05,0.05);
\node[below left] at (2.5+0,0){$a_3$};
\node[above left] at (2.5+0,0.5){$b_3$};
\node[below right] at (2.5+0.5,0){$a_4$};
\node[above right] at (2.5+0.5,0.5){$b_4$};

\draw[fill] (3.75+0,0) circle (0.03);
\draw[fill] (3.75+0.5,0) circle (0.03);
\draw[fill] (3.75+0.5,0.5) circle (0.03);
\draw[fill] (3.75+0,0.5) circle (0.03);
\draw[<-] (3.75+0,0.4)--(3.75+0,0);
\draw[dashed,->] (3.75+0,0.5)--(3.75+0.45,0.05);
\draw[<-] (3.75+0.5,0.1)--(3.75+0.5,0.5);
\draw[dashed,->] (3.75+0.5,0.5)--(3.75+0.05,0.05);
\node[below left] at (3.75+0,0){$c_3$};
\node[above left] at (3.75+0,0.5){$d_3$};
\node[below right] at (3.75+0.5,0){$d_4$};
\node[above right] at (3.75+0.5,0.5){$c_4$};
\node[above] at (3.125+0.25,0.8){$x_3=1$, $x_4=0$};

\node [below] at (0.875,-0.3){$w_1=1$};
\node [below] at (2.5+0.875,-0.3){$w_2=0$};
\end{tikzpicture}
\caption{Example of reduction from (Noisy) Boolean Hidden Matching to Maximum Acyclic Subgraph. 
Solid lines added by Alice, dashed lines added by Bob, and purple lines represent a cycle.  
Here we have $m_1=(1,2)$ and $m_2=(3,4)$. 
Note that $x_1\oplus x_2\oplus w_1=0$ has a cycle while $x_3\oplus x_4\oplus w_2=1$ has no cycles.}
\figlab{fig:mas:reduction}
\end{figure}

\thmmas*
\begin{proof}
Recall that the key property of the construction was that for each $i\in[2n]$, the subgraph corresponding to $m_i$ has eight edges, but consists of exactly one cycle when $x_{y_i}\oplus x_{z_i}\oplus w_i=0$ and zero cycles when $x_{y_i}\oplus x_{z_i}\oplus w_i=1$. 
Moreover, any subgraph corresponding to an unmatched vertex (of which there are $n$) will contribute two edges due to Alice's construction. 
Thus for the \YES case, when $Mx\oplus w=0^{\alpha n}$, the maximum acyclic subgraph has size 
\[7\cdot\frac{n}{2}+2\cdot n=\frac{11n}{2}.\]
In the \NO case, the probability that at least $\frac{np}{4}$ matching labels $w_i$ are flipped is at least $\frac{99}{100}$, since $p\ge\frac{128}{n}$ 
Therefore, in this case the maximum acyclic subgraph has size 
\[7\cdot\left(\frac{n}{2}-\frac{np}{4}\right)+8\cdot\frac{np}{4}+2\cdot n=\frac{11n}{2}+\frac{np}{4},\]
with probability at least $\frac{99}{100}$. 
For $p\le 1$, it follows that
\[\frac{11n/2+np/4}{11n/2}=1+\frac{p}{22}.\]
Hence, any $\left(1+\frac{p}{22}\right)$-approximation to the size of the maximum acyclic subgraph requires $\Omega\left(\sqrt{\frac{n}{p}}\right)$ space. 
\end{proof}

\subsection{Streaming Binary Tree Classification}
\seclab{sec:apps:tree}
In this section, we study the hardness of tree classification on insertion-only streams. 
In this setting, a stream of edge-insertions in an underlying graph with $n$ vertices arrive sequentially and the task is to classify whether the resulting graph is isomorphic to a complete binary tree on $n$ vertices, or $\delta$-far from one. 
\begin{problem}[Graph Classification]
Given a stream for the set of edges between $n$ vertices, determine whether the resulting graph induced by the stream is isomorphic to a complete binary tree on $n$ vertices. 
Here we assume that the number of vertices is $2^k-1$ for some integer $k$.
\end{problem}

\cite{EsfandiariHLMO18} used the following construction to show hardness of approximation for maximum matching size in the streaming model. 
Given an instance of Boolean Hidden Matching with an input vector $x\in\{0,1\}^{2n}$, a binary string $w\in\{0,1\}^{n}$, and a matching $M=\{(m_1,m'_1),(m_2,m'_2),\ldots,(m_n,m'_n)\}$ of $[2n]$, let $G$ be a graph with $6n$ vertices. 
Each bit $x_i$ is associated with vertices $v_i, v_{i,0}, v_{i,1}$ in $G$. 
Alice connects vertex $v_i$ to $v_{i,x_i}$ in $G$, e.g., if $v_i=0$ then Alice connects $v_i$ to $v_{i,0}$ and if $v_i=1$ then Alice connects $v_i$ to $v_{i,1}$. 
If $w_i=0$, then Bob creates an edge between $v_{m_i,0}$ and $v_{m'_i,1}$, as well as an edge between $v_{m_i,1}$ and $v_{m'_i,0}$. 
Otherwise if $w_i=1$, then Bob creates an edge between $v_{m_i,0}$ and $v_{m'_i,0}$, as well as an edge between $v_{m_i,1}$ and $v_{m'_i,1}$. 
Under this construction, if $x_{m_i}\oplus x_{m'_i}\oplus w_i=1$, then the six vertices $v_{m_i},v_{m_i,0},v_{m_i,1},v_{m'_i},v_{m'_i,0},v_{m'_i,1}$ form a path of length one and a path of length three.
Otherwise, if $x_{m_i}\oplus x_{m'_i}\oplus w_i=0$, then the six vertices form two paths of length two. 
We call this the EHLMO construction. 
See \figref{fig:bhh:paths} for an illustration. 

\begin{figure*}
\centering
\begin{subfigure}[b]{0.4\textwidth}
\begin{tikzpicture}
\draw[fill=black] (0,0) circle (0.1);
\node at (-0.5,0){$v_2$};
\draw[fill=black] (1,1) circle (0.1);
\node at (1,1.5){$v_{2,0}$};
\draw[fill=black] (1,-1) circle (0.1);
\node at (1,-1.5){$v_{2,1}$};
\draw[fill=black] (4,0) circle (0.1);
\node at (4.5,0){$v_4$};
\draw[fill=black] (3,1) circle (0.1);
\node at (3,1.5){$v_{4,0}$};
\draw[fill=black] (3,-1) circle (0.1);
\node at (3,-1.5){$v_{4,1}$};
\draw (0,0)--(1,1);
\draw[dashed] (1,1)--(3,-1);
\draw[dashed] (1,-1)--(3,1);
\draw (4,0)--(3,-1);
\end{tikzpicture}
\caption{Suppose $(m_4,m'_4)=(2,3)$, $x_2=0$, $x_3=1$, and $w_4=0$. Then $x_2\oplus x_3\oplus w_4=1$ and the construction forms a path of length one and a path of length three. Solid lines added by Alice, dashed lines added by Bob.}
\figlab{fig:bhh:zero}
\end{subfigure}
\qquad
\begin{subfigure}[b]{0.4\textwidth}
\begin{tikzpicture}
\draw[fill=black] (0,0) circle (0.1);
\node at (-0.5,0){$v_5$};
\draw[fill=black] (1,1) circle (0.1);
\node at (1,1.5){$v_{5,0}$};
\draw[fill=black] (1,-1) circle (0.1);
\node at (1,-1.5){$v_{5,1}$};
\draw[fill=black] (4,0) circle (0.1);
\node at (4.5,0){$v_7$};
\draw[fill=black] (3,1) circle (0.1);
\node at (3,1.5){$v_{7,0}$};
\draw[fill=black] (3,-1) circle (0.1);
\node at (3,-1.5){$v_{7,1}$};
\draw (0,0)--(1,-1);
\draw[dashed] (1,1)--(3,1);
\draw[dashed] (1,-1)--(3,-1);
\draw (4,0)--(3,1);
\end{tikzpicture}
\caption{Suppose $(m_6,m'_6)=(5,7)$, $x_5=1$, $x_7=0$, and $w_6=1$. Then $x_5\oplus x_7\oplus w_6=0$ and the construction forms two paths of length two. Solid lines added by Alice, dashed lines added by Bob.}
\figlab{fig:bhh:one}
\end{subfigure}
\caption{Construction of graph from instance of Boolean Hidden Matching. 
Observe that in the case $Mx\oplus w=1^n$, $v_{5,1}$ and $v_{7,0}$ can be embedded into the bottom layer of a complete binary tree to form another complete binary tree, while in the $Mx\oplus w=0^n$ case, $v_{2,0}$ and $v_{4,1}$ will induce a cycle in the graph.}
\figlab{fig:bhh:paths}
\end{figure*}
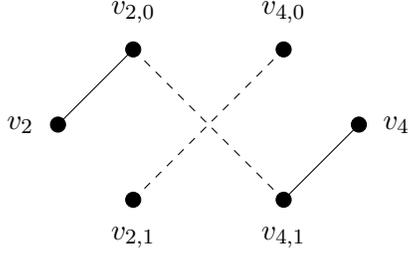
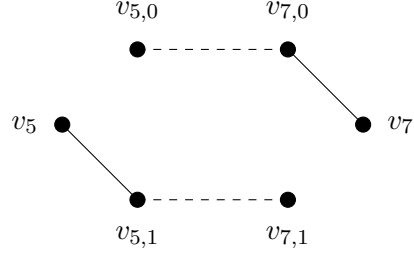

To show hardness of classifying whether an underlying graph is a complete binary tree, we use the EHLMO construction to embed an instance of $p$-Noisy BHM of size $2n$ into the bottom layer of a binary tree with $4n-1$ total nodes. 
For each $i\in[2n]$, the case $(Mx)_i\oplus w_i=0$ creates two paths of length two. 
We use these two paths of length two to extend the binary tree to an additional layer at two different nodes. 
On the other hand, the case $(Mx)_i\oplus w_i=1$ creates a path of length one and a path of length three. 
By using the same construction, the case $(Mx)_i\oplus w_i=1$ thus results in a non-root node in the tree having degree two; hence the resulting graph is a complete binary tree if and only if $Mx\oplus w=0^n$. 

Formally, let $n$ be a power of $2$ and let $T$ be a complete binary tree with $2n$ leaves. 
Consider the EHLMO construction to embed an instance of the above Noisy Boolean Hidden Matching problem with noise $p=\frac{4}{n}$ into the bottom layer of a tree, so that in the case where $Mx\oplus w=0^n$, the resulting graph is a complete binary tree, while in the case $Mx\oplus w\neq0^n$, the resulting graph contains a cycle.
More formally:
\begin{itemize}
\item
Alice creates a complete $6n$ vertices $v_i, v_{i,0}, v_{i,1}$ for each $i\in[2n]$.
\item
Alice creates a complete binary tree with leaves $v_{i,x_i}$ for $i\in[2n]$ (and new non-leaf vertices, distinct from any previously created).
\item
Alice also creates an edge between vertices $v_i$ and $v_{i,x_i}$ for each $i\in[2n]$. 
\item
For each edge $m_i=(y_i,z_i)$ of the matching $M$ with $i\in[n]$, Bob creates an edge between $v_{y_i,0}$ and $v_{z_i,w_i\oplus 1}$ as well as an edge between $v_{y_i,1}$ and $v_{z_i,w_i}$. 
\end{itemize}
\thmtreenoisybhm*
\begin{proof}
We claim that the resulting construction is a complete binary tree in the case where $Mx\oplus w=0^n$ and is not a complete binary tree in the case where $Mx\oplus w\neq0^n$, which happens with high constant probability in the \NO case of the underlying $p$-Noisy BHM instance. 
We first consider the \YES case, i.e., $Mx\oplus w=0^n$, so that $x_{y_i}\oplus x_{z_i}\oplus w_i=0$ for all $i\in[n]$, where we recall that the edge $m_i=(y_i,z_i)$. 
Observe that if $w_i=0$ and $x_{y_i}\oplus x_{z_i}\oplus w_i=0$, then $x_{y_i}\neq x_{z_i}$. 
Since Bob connects $v_{y_i,0}$ with $v_{z_i,0}$ and $v_{y_i,1}$ with $v_{z_i,1}$ if $w_i=1$, then $v_{y_i,x_{y_i}}$ will be the middle node of a path of length two. 
Similarly, $v_{z_i,x_{z_i}}$ will be the middle node of a path of length two. 
Hence, the resulting graph will be a complete binary tree. 

In the \NO case, each $i\in[n]$ has $x_{y_i}\oplus x_{z_i}\oplus w_i=1$ with probability $\frac{4}{n}$. 
Let $\mathcal{E}$ be the event that we are in the \NO case and there exists some $i\in[n]$ such that $x_{y_i}\oplus x_{z_i}\oplus w_i=1$. 
Since we are in the \NO case with probability $\frac{1}{2}$, then 
\[\PPr{\mathcal{E}}=\frac{1}{2}\left(1-\left(1-\frac{4}{n}\right)^n\right)>\frac{1}{2}\left(1-\frac{1}{e^4}\right)\ge\frac{15}{32}.\]
Conditioned on $\mathcal{E}$, Bob connects $v_{y_i,0}$ with $v_{z_i,1}$ and $v_{y_i,1}$ with $v_{z_i,0}$, so there exists an edge between $v_{y_i,x_{y_i}}$ and $v_{z_i,x_{z_i}}$ and the resulting graph contains a cycle. 

Hence, any algorithm $\mathcal{A}$ that classifies whether the underlying graph is a complete binary tree with probability $\frac{3}{4}$ also solves the Noisy Boolean Hidden Matching problem with probability at least $\frac{1}{2}\frac{3}{4}\cdot\frac{1}{2}\cdot\frac{15}{32}=\frac{39}{64}>\frac{2}{3}$, where the first term in the summand represents the probability of success in the \YES case and the second term in the summand represents the probability of success in the \NO case.
Thus by \thmref{thm:pbhh} with $t=2$, $\mathcal{A}$ uses $\Omega\left(\sqrt{\frac{n}{p}}\right)=\Omega(n)$ space. 
\end{proof}

We note that~\thmref{thm:tree:noisy:bhm} also follows by a reduction from the standard INDEX communication problem. The following more general \thmref{thm:tree:noisy:bhm:pt}, however, does not. It follows by the same reduction we use for~\thmref{thm:tree:noisy:bhm:pt} and follows from a nearly identical proof, with the only change being that the variable $p$ is set to $\epsilon$ in the reduction from $p$-Noisy BHM.

\thmtreenoisybhmpt*
\begin{proof}
We use the identical reduction to the proof of~\thmref{thm:tree:noisy:bhm}, with the parameter $p$ set to $\epsilon$ instead of $\tfrac4n$. It is still true that in the \YES case the resulting graph is isomorphic to a complete binary tree. In the \NO case, of the $n$ edge labels, we expect $pn=\epsilon n$ of edge labels to be flipped. In fact, since $\epsilon>\tfrac{64}n$, at least $\epsilon n/2$ edges are flipped with probability at least $99/100$. This results in $\epsilon n/2$ disjoint cycles in the underlying graph. As a result, the underlying graph requires at least $\epsilon n/2$ edge insertions or deletions to be transformed into a complete binary tree (or in fact a tree of any kind). Since $|V|=8n-1$, this means the underlying graph in the \NO case is at least $\epsilon/16$ far from being a complete binary tree.

The statement of the theorem then follows by~\thmref{thm:pbhh}.
\end{proof}

\section*{Acknowledgements}
Michael Kapralov and Jakab Tardos would like to thank support from ERC Starting Grant 759471. 
David P. Woodruff and Samson Zhou would like to thank support from NSF grant No. CCF-181584 and a Simons Investigator Award. 

\def\shortbib{0}
\bibliographystyle{alpha}
\bibliography{references}

\end{document}